\documentclass[10pt]{article} 

\usepackage[margin=1.1in]{geometry}
\usepackage{amsmath,amssymb,amsfonts,amsthm}
\usepackage{graphicx}
\usepackage{float}
\usepackage{mathtools}
\usepackage{algorithm}
\usepackage[noend]{algpseudocode}
\usepackage{xcolor}
\usepackage{tikz}
\usetikzlibrary{arrows.meta,positioning,calc,fit,backgrounds}
\usepackage[authoryear, round]{natbib}
\usepackage{booktabs}
\usepackage{hyperref}
\usepackage[capitalize,nameinlink]{cleveref}
\usepackage{titlesec}
\usepackage{fancyhdr}
\usepackage{microtype}
\usepackage{parskip}

\newtheorem{theorem}{Theorem}
\newtheorem{proposition}{Proposition}
\newtheorem{corollary}{Corollary}
\newtheorem{lemma}{Lemma}

\newtheorem{remark}{Remark}

\crefname{algorithm}{Algorithm}{Algorithms}

\hypersetup{
  colorlinks=true,
  linkcolor=blue!70!black,
  citecolor=blue!50!black,
  urlcolor=blue!60!black
}

\providecommand{\E}{\mathbb{E}}
\providecommand{\Prob}{\mathbb{P}}
\newcommand{\R}{\mathbb R}

\newcommand{\cF}{\mathcal F}

\newcommand{\cG}{\mathcal G}
\newcommand{\cX}{\mathcal X}
\newcommand{\cY}{\mathcal Y}
\newcommand{\Dtr}{\mathcal D_{\mathrm{tr}}}

\newcommand{\pzeroD}[2]{p_0(#1\mid #2,\Dtr)}
\newcommand{\phD}[2]{p_h(#1\mid #2,\Dtr)}
\newcommand{\ZhD}[1]{Z_h(#1,\Dtr)}
\newcommand{\pzeroDobj}{p_0(\cdot\mid\cdot,\Dtr)}
\newcommand{\phDobj}{p_h(\cdot\mid\cdot,\Dtr)}
\newcommand{\NLPD}{\mathrm{NLPD}}
\newcommand{\KL}{D_{\mathrm{KL}}}
\newcommand{\defeq}{\coloneqq}
\newcommand{\Var}{\mathrm{Var}}

\titleformat{\section}{\large\bfseries\color{blue!60!black}}{{\thesection}}{1em}{}[\titlerule]
\titleformat{\subsection}{\normalsize\bfseries\color{blue!40!black}}{{\thesubsection}}{1em}{}
\titleformat{\subsubsection}{\normalsize\bfseries\color{blue!40!black}}{{\thesubsubsection}}{1em}{}

\begin{document}

\begin{center}
  {\LARGE\bfseries \textsf{Anytime-Valid Evidence for Prespecified\\[5pt]
  Predictive Corrections}}\\[2em]
  {\large Seungjin Choi}\\[1em]
  {\normalsize CROID Research and aSSIST University, Seoul, Korea}
  \end{center}

\vspace{0.5em}
\noindent\rule{\linewidth}{1.5pt}
\vspace{0.5em}

\begin{abstract}
A predictive correction is a prespecified modification of an existing
predictive distribution intended to reflect an anticipated change in future
outcomes given their inputs, motivated, for example, by instrument
recalibration, assay drift, or a known intervention. We study how to
accumulate anytime-valid evidence that such a correction predicts incoming
target outcomes better than the uncorrected source predictive distribution.
A fixed nonnegative tilt transforms the source predictive into a corrected
predictive, and the corrected-to-source predictive likelihood ratio is a
conditional e-value whose running product forms an e-process. This process
remains valid under optional stopping and arbitrary input sequences,
including adaptively selected ones, while its logarithm equals the cumulative
predictive log-score advantage of the correction. A conditional drift
decomposition characterizes evidence growth under an arbitrary target
predictive distribution, and a correction-dependent half-space identifies
misspecified target distributions for which the same false-confirmation bound
continues to hold. When the predictive likelihood ratio is strictly positive,
its reciprocal yields an anytime-valid refutation boundary, while an
overshoot identity explains why the realized null crossing probability may
fall below the nominal level. Label-shift, conditional mean and variance,
subgroup-specific, and exponential-family corrections arise as special cases.
Prespecified mixtures accommodate uncertainty over corrections, predictable
tilts permit adaptive betting, and beyond-tolerance comparisons target changes
large enough to justify action. Cross-family calculations and synthetic
experiments show that a boundary crossing supports the proposed correction
relative to its reference but does not uniquely identify the mechanism
responsible for the shift.
\end{abstract}

\newpage
\tableofcontents
\newpage

\section{Introduction}
\label{sec:intro}

A predictive correction is a prespecified modification of an existing predictive distribution intended to reflect an anticipated change in future outcomes given their inputs. Such a correction may be motivated by scientific knowledge, engineering analysis, or an operational policy before the outcomes used to evaluate it are observed. Distribution shift is more commonly treated as an estimation or adaptation problem: target data are used to identify what has changed and to learn an appropriate modification of the source model. This approach is natural when target data are plentiful and the shift is sufficiently identifiable. In small-batch scientific and operational settings, however, target outcomes may arrive sequentially, and a plausible correction may already be available before monitoring begins. The immediate question is then not how to estimate an unrestricted target distribution, but whether the proposed correction predicts the incoming outcomes better than retaining the original source predictive distribution.

We study this complementary problem of \emph{confirmation}. Let $\Dtr$ denote the source training data, and let $\pzeroD{y}{x}$ be a fixed source predictive distribution for an outcome $y$ at input $x$, conditional on $\Dtr$. Before observing the outcomes used for confirmation, the practitioner specifies a nonnegative tilt
\[
  h:\cX\times\cY\to[0,\infty)
\]
with finite and positive normalizer
\[
  0<\ZhD{x}
  =
  \int h(x,y)\pzeroD{y}{x}\,dy
  <\infty
\]
for every relevant input $x$. The tilt defines the corrected predictive distribution
\begin{equation}
  \phD{y}{x}
  =
  \frac{h(x,y)\pzeroD{y}{x}}{\ZhD{x}},
  \qquad
  \ZhD{x}
  =
  \int h(x,y)\pzeroD{y}{x}\,dy.
  \label{eq:tilted-general}
\end{equation}
In the primary setting, $h$, and hence $\phDobj$, is fixed before the testing outcomes are observed. The incoming pairs $(X_i,Y_i)$ are then used only to evaluate the proposed correction. Confirmation does not estimate the full target distribution, prove that the correction is exactly specified, or identify the mechanism responsible for the shift. It provides sequential evidence that the corrected predictive outpredicts the source predictive on the observed target stream.

Predictive corrections of this form arise naturally in applications. Calibration transfer or knowledge of a changed instrument may suggest a systematic modification of the predicted response \citep{WorkmanJJ2018as}. A new laboratory batch or assay protocol may suggest a change in conditional variability \citep{JohnsonWE2007biostatistics,LeekJT2010nrg}. A discrepancy model may relate simulator output to anticipated physical observations \citep{KennedyMC2001jrsssb}, while an operational policy may distinguish acceptable degradation from a change large enough to require intervention \citep{PodkopaevA2022iclr}. These examples share a crucial feature: the form of the correction is motivated independently of the outcomes subsequently used to confirm it. We later allow tilts that are updated predictably using past observations, but their interpretation is different. They define adaptive betting strategies against the source predictive null rather than confirmation of one fixed prespecified correction.

The main construction follows from a conditional likelihood-ratio argument. Let $\cF_{i-1}$ denote the information available before observing the $i$th outcome, including $\Dtr$ and the previous testing pairs. Under the source predictive null,
\begin{equation}
  H_0^{\mathrm{pred}}:\quad
  Y_i\mid X_i,\cF_{i-1}
  \sim
  \pzeroD{\cdot}{X_i}.
\end{equation}
Let $\mathcal P_0^{\mathrm{pred}}$ denote the class of all data-stream
distributions satisfying this conditional null, with the admissible input
mechanism left unrestricted. Thus, the input $X_i$ may be stochastic or
deterministic, depend on the past, or be selected by an adaptive experimental
design. Define the one-step corrected-to-source likelihood ratio
\begin{equation}
  e_i
  =
  \frac{\phD{Y_i}{X_i}}
       {\pzeroD{Y_i}{X_i}}
  =
  \frac{h(X_i,Y_i)}{\ZhD{X_i}}.
  \label{eq:e-general}
\end{equation}
Conditionally on $(\cF_{i-1},X_i)$, this ratio has expectation one under $H_0^{\mathrm{pred}}$. Thus, $e_i$ is a conditional e-value, and its running product
\[
  M_t=\prod_{i=1}^t e_i
\]
is a nonnegative martingale under the source predictive null. Ville's inequality \citep{VilleJ1939phd} therefore gives, for every $\alpha\in(0,1)$,
\begin{equation}
  \sup_{P\in\mathcal P_0^{\mathrm{pred}}}
  \Prob_P\!\left\{
    \sup_{t\ge0}M_t>\frac1\alpha
  \right\}
  \le\alpha.
  \label{eq:ville-intro}
\end{equation}
Consequently, the stopping rule
\[
  \tau^*
  =
  \inf\left\{
    t\ge1:M_t>\frac1\alpha
  \right\}
\]
controls the probability of falsely confirming the correction under continuous monitoring and optional stopping. This guarantee holds for arbitrary input sequences because validity is established conditionally on each realized input.

With the exact normalizer, $M_t$ is the sequential likelihood ratio between the source and corrected conditional predictive distributions. The procedure is therefore closely related to one-sided likelihood-ratio monitoring in Wald's sequential probability ratio framework \citep{WaldA1945aoms}, but the e-process formulation emphasizes continuously reportable evidence, optional-stopping validity, and compatibility with adaptive input selection. Moreover,
\[
  \log M_t
  =
  \sum_{i=1}^t
  \left\{
    \log\phD{Y_i}{X_i}
    -
    \log\pzeroD{Y_i}{X_i}
  \right\}
\]
is exactly the cumulative predictive log-score advantage of the corrected predictive over the source predictive. When $h(x,y)>0$ $\pzeroD{\cdot}{x}$-almost surely, the reciprocal likelihood ratio yields a corresponding anytime-valid boundary for refuting the corrected predictive in favor of the source predictive.

The construction generalizes anytime-valid confirmation of label-shift corrections \citep{Choi2026testing}. Under label shift, the tilt has the restricted form $h(x,y)=w(y)$. Allowing $h$ to depend jointly on $x$ and $y$ covers conditional mean and variance corrections, subgroup-specific corrections, and general exponential-family predictive tilts. This generality also clarifies the scope of the resulting evidence. Because \eqref{eq:e-general} conditions on the realized input, the process is deliberately insensitive to pure covariate shift when the conditional distribution $Y\mid X$ remains unchanged. Confirmation of a proposed covariate correction or of covariate balance is instead a problem concerning the marginal input distribution and is studied separately in \citet{Choi2026arxiv_avccs}.

The guarantee in \eqref{eq:ville-intro} is a statement about the source predictive null. A separate robustness question arises when the actual target predictive distribution is neither $\pzeroDobj$ nor $\phDobj$. Such a target does not satisfy the original null. Nevertheless, we show that the same false-confirmation bound continues to hold over a correction-dependent half-space of misspecified target predictive distributions characterized by a conditional moment inequality. This result does not enlarge or redefine the null; it identifies target misspecifications under which the original directional confirmation rule remains controlled. Outside this protected set, the correction may acquire positive drift under a structurally different target, or rare large e-values may increase the crossing probability even when the average log-growth is negative. This distinction is central to the interpretation of the method: a crossing supports the proposed correction relative to its reference predictive but does not uniquely identify the form or cause of the underlying distribution shift.

The main contributions are as follows:
\begin{itemize}

\item
\textbf{Anytime-valid relative predictive evidence.}
We construct a conditional e-process for a prespecified predictive correction. Its wealth is exactly the cumulative predictive likelihood ratio, and its log wealth is the cumulative log-score advantage of the corrected predictive over the source predictive. Validity holds under optional stopping and arbitrary input sequences, including adaptively selected ones.

\item
\textbf{Evidence growth and finite-horizon behavior.}
We derive an inputwise conditional drift decomposition and characterize long-run growth through differences in conditional Kullback--Leibler divergences. We also provide finite-horizon bounds on the probability that a positively drifting process has not yet crossed its confirmation boundary under i.i.d.\ sampling.

\item
\textbf{False-confirmation control under target misspecification.}
We identify a correction-dependent half-space of misspecified target predictive distributions under which the original e-process retains the same anytime-valid bound on the confirmation boundary. This geometry explains why generic proximity to the source in Kullback--Leibler, total-variation, or Hellinger distance does not by itself preserve the bound, and it separates automatically protected misspecifications from unsupported robustness claims.

\item
\textbf{Confirmation, refutation, and operational extensions.}
Strictly positive likelihood ratios support both an upper boundary for confirming the correction and a reciprocal boundary for refuting it. An overshoot identity explains why the realized null crossing probability may be below the nominal level. We further develop predictable tilts, prespecified mixtures and correction panels, and beyond-tolerance comparisons that control false confirmation over an entire tolerated region in regular one-parameter exponential families.

\item
\textbf{Structured corrections and cross-family diagnostics.}
We derive label-shift, conditional mean, conditional variance, subgroup-specific, and general exponential-family corrections as special cases. Analytic calculations and synthetic experiments distinguish within-family magnitude mismatch from cross-family structural mismatch, show how an unintended mechanism can generate evidence for a proposed correction, and clarify that relative predictive confirmation is not mechanism identification.

\end{itemize}

\section{Related Work}
\label{sec:related}

\paragraph{E-values, anytime-valid inference, and sequential likelihood ratios.}
E-values are nonnegative evidence measures, and predictable products of conditional e-values form e-processes that remain valid under optional stopping \citep{VovkV2021aos,ShaferG2021jrsssa,RamdasA2023ss}. Their time-uniform guarantee follows from Ville's inequality \citep{VilleJ1939phd}, while testing by betting and game-theoretic probability connect these ideas to martingales and prequential prediction \citep{DawidAP1984jrsssa,VovkV2005book,ShaferG2019book}. For a fixed source predictive and a fixed corrected predictive, the process studied here is a sequential likelihood ratio, and the two-boundary rule of \cref{sec:refutation} is closely related to Wald's sequential probability ratio test \citep{WaldA1945aoms}. Our contribution is therefore not a new likelihood-ratio test for a simple pair. It is to use a practitioner-specified predictive correction as the alternative, retain conditional validity under arbitrary and adaptively selected input sequences, and characterize the resulting evidence growth, robustness under target misspecification, and operational extensions.

\paragraph{Distribution shift, adaptation, and correction confirmation.}
Distribution shift includes covariate shift, label shift, concept shift, and more general joint shift \citep{Quinonero-CandelaJ2009book,SugiyamaM2012book_a}. Most methods estimate the target shift or adapt a source model using labeled or unlabeled target data. For example, label-shift methods estimate target class proportions using source classifiers or calibrated predictors \citep{LiptonZ2018icml,AlexandariAM2020icml,GargS2020neurips}. The task considered here is different: the correction is specified before the confirming outcomes are observed, and those outcomes are used to accumulate evidence for or against that correction rather than to estimate an unrestricted target distribution. \citet{Choi2026testing} developed this confirmation perspective for prespecified label-shift corrections. The present paper extends it from label-only tilts $h(x,y)=w(y)$ to corrections that may depend jointly on inputs and outcomes. Covariate-shift methods instead concern density ratios over the marginal input distribution \citep{SugiyamaM2012book_a}; anytime-valid confirmation of a proposed covariate correction and of covariate balance is treated separately in \citet{Choi2026arxiv_avccs}.

\paragraph{Sequential model monitoring and tolerated change.}
Sequential monitoring of deployed models is often framed as testing whether a risk, loss, or performance functional has crossed an unacceptable level. Anytime-valid procedures for monitoring such scalar functionals have been developed for deployment settings in which acceptable risk levels are specified in advance \citep{PodkopaevA2022iclr}. Our target is different: we compare two full conditional predictive distributions, namely a source predictive and a prespecified corrected predictive. The beyond-tolerance construction in \cref{sec:beyond-tolerance-correction} is operationally related to risk-threshold monitoring, but it compares an actionable predictive directly with a tolerated-boundary predictive and, within a regular one-parameter exponential family, controls false confirmation over the entire tolerated region.

\paragraph{Predictive scoring, calibration, and conformal prediction.}
The logarithm of the likelihood-ratio e-process is a cumulative difference in predictive log scores, linking the procedure to prequential evaluation of probabilistic forecasts \citep{DawidAP1984jrsssa}. The e-process adds an inferential guarantee to that comparison: under the source predictive null, the evidence can be monitored continuously without invalidating the error bound. This objective differs from predictive calibration and coverage. Conformal prediction provides finite-sample marginal coverage under exchangeability and has been adapted to covariate and label shift through weighted calibration \citep{VovkV2005book,TibshiraniR2019neurips,PodkopaevA2021uai,AngelopoulosAN2023ftml}. Conformal Bayes combines Bayesian predictive information with conformal calibration to obtain finite-sample marginal coverage without requiring the Bayesian predictive model to be correctly specified \citep{FongE2021neurips}. Under label shift, \citet{Choi2026eiml,Choi2026arxiv_scbc} use predictive tilting and weighted calibration to adapt conformal Bayes prediction sets. Those methods target prediction-set coverage or calibration, whereas the present paper uses the corrected-to-source predictive ratio to accumulate anytime-valid evidence for a proposed correction.

\section{General Predictive-Correction E-Process}
\label{sec:general}

In this section, we develop the general framework for evaluating a prespecified
predictive correction as target outcomes are observed sequentially.  We first
show that normalization of the correction tilt produces a corrected predictive
distribution whose ratio to the source predictive is a conditional e-value.
The resulting product e-process provides anytime-valid \emph{relative
confirmation}: a boundary crossing favors the corrected predictive over the
source predictive without estimating the full target distribution or
identifying the mechanism responsible for the change.  We then characterize
evidence growth under arbitrary target predictive distributions and study the
robustness of false-confirmation control under target misspecification by
identifying a correction-dependent protected class.  Finally, we develop reciprocal
refutation, overshoot accounting, safe numerical normalization, and predictable
corrections based on past observations.

\subsection{Problem Setup}
\label{sec:setup}

Let $\Dtr$ denote the source training data, and let $\pzeroD{y}{x}$ be a fixed
source predictive distribution for an outcome $y$ at input $x$, conditional on
$\Dtr$.  The conditioning on $\Dtr$ includes all model fitting, posterior
updating, calibration, and other training-stage operations completed before
monitoring begins.  A target stream consists of input--outcome pairs
\[
  (X_1,Y_1),(X_2,Y_2),\ldots.
\]

Let $\sigma(Z_1,\ldots,Z_k)$ denote the $\sigma$-algebra generated by the
random quantities $Z_1,\ldots,Z_k$; it represents all information that can be
determined from their observed values.  Define
\[
  \cF_0=\sigma(\Dtr)
\]
and, for $t\ge1$,
\[
  \cF_t
  =
  \sigma\!\left(
    \Dtr,X_1,Y_1,\ldots,X_t,Y_t
  \right).
\]
Thus, $\cF_t$ contains all information available after the first $t$ target
input--outcome pairs have been observed.  Before observing $Y_i$, define
\[
  \cG_i
  =
  \cF_{i-1}\vee\sigma(X_i),
\]
where $\vee$ denotes the smallest $\sigma$-algebra containing both
$\cF_{i-1}$ and $\sigma(X_i)$.  Hence, $\cG_i$ contains the past and the
current input $X_i$, but not its corresponding outcome $Y_i$.

We condition throughout on the realized training data $\Dtr$, equivalently
treating it as part of the initial $\sigma$-field.  All conditional densities
are defined with respect to a common dominating measure on $\cY$; the same
notation covers discrete outcomes, with integrals replaced by sums.  When a
conditional distribution $r_i(\cdot\mid X_i)$ is determined by the information
in $\cG_i$, the notation
\[
  \E_{Y\sim r_i(\cdot\mid X_i)}
  \bigl[u(X_i,Y)\bigr]
\]
means expectation with respect to that conditional distribution.  Equivalently,
under the specification
\[
  Y_i\mid\cG_i\sim r_i(\cdot\mid X_i),
\]
it denotes a version of
\[
  \E\!\left[u(X_i,Y_i)\mid\cG_i\right].
\]

The source predictive null is
\begin{equation}
  H_0^{\mathrm{pred}}:
  \qquad
  Y_i\mid\cG_i
  \sim
  \pzeroD{\cdot}{X_i}
  \quad\text{for every }i.
  \label{eq:section3-predictive-null}
\end{equation}
Let $\mathcal P_0^{\mathrm{pred}}$ denote the class of all data-stream
distributions satisfying \eqref{eq:section3-predictive-null}.  This class is
composite because the null specifies only the conditional distribution of
$Y_i$ given $\cG_i$ and leaves the input mechanism unrestricted.  The inputs
may be deterministic or stochastic, dependent on the past, or selected by an
adaptive experimental-design rule.

A nonnegative process $(E_t)_{t\ge0}$, adapted to $(\cF_t)_{t\ge0}$ and
initialized at $E_0=1$, is an \emph{e-process} for
$\mathcal P_0^{\mathrm{pred}}$ if
\begin{equation}
  \sup_{P\in\mathcal P_0^{\mathrm{pred}}}
  \E_P[E_\tau]
  \le1
  \label{eq:eprocess-definition}
\end{equation}
for every stopping time $\tau$.  For a possibly infinite stopping time, we use
the convention
\[
  E_\tau=E_\infty
  \defeq
  \liminf_{t\to\infty}E_t
  \qquad\text{on }\{\tau=\infty\}.
\]
Thus, continuous monitoring and data-dependent stopping do not increase the
expected evidence above one under any distribution in the null class.  Every
nonnegative supermartingale with initial value one is an e-process: apply
optional stopping to $\tau\wedge t$ and then use Fatou's lemma as
$t\to\infty$.  Exact normalization will make the primary wealth process below
a martingale under every $P\in\mathcal P_0^{\mathrm{pred}}$.

A predictive correction is specified by a jointly measurable nonnegative tilt
\[
  h:\cX\times\cY\to[0,\infty)
\]
with finite and positive normalizer
\begin{equation}
  0<\ZhD{x}
  =
  \int h(x,y)\pzeroD{y}{x}\,dy
  <\infty
  \label{eq:normalizer}
\end{equation}
for every relevant input $x$.  We assume that the source predictive and the
tilt are measurable so that $x\mapsto\ZhD{x}$ is measurable.  The tilt defines
the corrected predictive $\phDobj$ in \eqref{eq:tilted-general}.  In the
primary setting, $h$, and hence $\phDobj$, is fixed before monitoring begins;
\cref{sec:predictable} later allows predictable updates based on past
observations.  The inferential object is the comparison between the corrected
and source predictive distributions, not estimation of the unknown target
distribution itself.

\subsection{Anytime-Valid Relative Confirmation of a Predictive Correction}
\label{sec:eprocess-construction}

Suppose that the practitioner has prespecified $\phDobj$ before observing the
target outcomes.  The operational question is whether the accumulating
outcomes provide sufficient evidence to reject continued use of
$\pzeroDobj$ in the direction represented by $\phDobj$.  The corrected
predictive determines the direction in which evidence against
$H_0^{\mathrm{pred}}$ is accumulated; it is not itself assumed to be the true
target predictive.

\bigskip
\begin{lemma}[Normalized tilt as a predictive likelihood ratio]
\label{lem:ratio}
For each \(x\) satisfying \eqref{eq:normalizer}, \(\phD{\cdot}{x}\) is a
probability distribution absolutely continuous with respect to
\(\pzeroD{\cdot}{x}\), and
\begin{equation}
  \frac{\phD{y}{x}}{\pzeroD{y}{x}}
  =
  \frac{h(x,y)}{\ZhD{x}}
  \qquad
  \pzeroD{\cdot}{x}\text{-a.s.}
  \label{eq:normalized-likelihood-ratio}
\end{equation}
Moreover,
\begin{equation}
  \E_{Y\sim\pzeroD{\cdot}{x}}
  \!\left[
    \frac{h(x,Y)}{\ZhD{x}}
  \right]
  =1.
  \label{eq:normalized-ratio-mean-one}
\end{equation}
\end{lemma}

\begin{proof}[Proof sketch]
Substituting the definition of \(\phD{\cdot}{x}\) and using
\eqref{eq:normalizer} gives unit integral, the likelihood-ratio identity, and
expectation one under \(\pzeroD{\cdot}{x}\).  See
\cref{app:proof-ratio} for details.
\end{proof}

\paragraph{Canonical evidence factor outside the source support.}
The likelihood-ratio identity in \cref{eq:normalized-likelihood-ratio} is an
$\pzeroD{\cdot}{x}$-almost-sure statement, as is standard for a
Radon--Nikodym derivative.  Throughout the paper we therefore take
\begin{equation}
  e(x,y)\defeq\frac{h(x,y)}{\ZhD{x}}
  \label{eq:canonical-e-factor}
\end{equation}
as the canonical measurable version of the one-step evidence factor.  Under
the source predictive null it coincides almost surely with the
corrected-to-source predictive likelihood ratio.  It remains well defined for
a target distribution that is not dominated by the source predictive.  The
literal predictive-likelihood-ratio, log-score, and KL interpretations below
are invoked only when the relevant densities and logarithms are well defined.

\bigskip
\begin{proposition}[Per-observation relative e-value]
\label{prop:evalue}
For the \(i\)th observation, define the canonical one-step factor
\begin{equation}
  e_i
  \defeq
  \frac{h(X_i,Y_i)}{\ZhD{X_i}}.
  \label{eq:one-step-evalue}
\end{equation}
Under the source predictive null, \cref{lem:ratio} gives
$e_i=\phD{Y_i}{X_i}/\pzeroD{Y_i}{X_i}$ almost surely.  For every
\(P\in\mathcal P_0^{\mathrm{pred}}\),
\begin{equation}
  \E_P[e_i\mid\cG_i]=1.
  \label{eq:conditional-evalue}
\end{equation}
Hence \(e_i\) is a conditional e-value for the source predictive null.
\end{proposition}

\begin{proof}[Proof sketch]
Condition on \(\cG_i\) and apply \cref{lem:ratio} under
\eqref{eq:section3-predictive-null}.  See
\cref{app:proof-evalue}.
\end{proof}

\bigskip
\begin{theorem}[Anytime-valid relative confirmation of a predictive correction]
\label{thm:eprocess}
Let \(M_0=1\) and
\begin{equation}
  M_t
  \defeq
  \prod_{i=1}^t e_i
  =
  \prod_{i=1}^t
  \frac{h(X_i,Y_i)}{\ZhD{X_i}}.
  \label{eq:relative-eprocess}
\end{equation}
Under every $P\in\mathcal P_0^{\mathrm{pred}}$, this process agrees almost
surely at every finite time with the corrected-to-source predictive
likelihood-ratio product.  For every \(P\in\mathcal P_0^{\mathrm{pred}}\),
\((M_t)_{t\ge0}\) is a nonnegative martingale with respect to
\((\cF_t)_{t\ge0}\).  Hence it is an e-process for
\(\mathcal P_0^{\mathrm{pred}}\), and for every \(\alpha\in(0,1)\),
\begin{equation}
  \sup_{P\in\mathcal P_0^{\mathrm{pred}}}
  \Prob_P\!\left\{
    \sup_{t\ge0}M_t>\frac{1}{\alpha}
  \right\}
  \le\alpha.
  \label{eq:anytime-relative-bound}
\end{equation}
Therefore, the stopping time
\begin{equation}
  \tau^*
  \defeq
  \inf\left\{
    t\ge1:M_t>\frac{1}{\alpha}
  \right\}
  \label{eq:relative-stopping-time}
\end{equation}
satisfies
\begin{equation}
  \sup_{P\in\mathcal P_0^{\mathrm{pred}}}
  \Prob_P(\tau^*<\infty)
  \le\alpha.
  \label{eq:relative-stopping-bound}
\end{equation}
The guarantee holds under continuous monitoring and for every admissible input
mechanism, including adaptive selection based on past observations.
\end{theorem}

\begin{proof}[Proof sketch]
Fix \(P\in\mathcal P_0^{\mathrm{pred}}\).  By
\cref{prop:evalue} and iterated conditional expectation,
\[
  \E_P[e_t\mid\cF_{t-1}]
  =
  \E_P\!\left[
    \E_P[e_t\mid\cG_t]
    \,\middle|\,
    \cF_{t-1}
  \right]
  =1.
\]
Therefore,
\(\E_P[M_t\mid\cF_{t-1}]=M_{t-1}\), so \((M_t)\) is a
nonnegative \(P\)-martingale.  Since this holds for every
\(P\in\mathcal P_0^{\mathrm{pred}}\), the process is an e-process for the
whole null class.  Ville's inequality gives
\eqref{eq:anytime-relative-bound}, and
\eqref{eq:relative-stopping-bound} follows from
\[
  \{\tau^*<\infty\}
  =
  \left\{
    \sup_{t\ge0}M_t>\frac1\alpha
  \right\}.
\]
See \cref{app:proof-eprocess}.
\end{proof}

\paragraph{Classical likelihood-ratio monitoring as a special case.}
The predictive-ratio framework includes ordinary conditional
likelihood-ratio monitoring.  If the source and corrected predictives are two
fixed, fully specified conditional likelihoods,
\[
  \pzeroD{y}{x}=f(y\mid x,\theta_0),
  \qquad
  \phD{y}{x}=f(y\mid x,\theta_1),
\]
then
\[
  e_i
  =
  \frac{f(Y_i\mid X_i,\theta_1)}
       {f(Y_i\mid X_i,\theta_0)}
\]
is the classical one-step likelihood ratio, and \(M_t\) is its sequential
product.  The predictive formulation is more general because it also permits
posterior predictive distributions, fitted predictive distributions treated
as fixed conditional on \(\Dtr\), and corrections specified directly at the
level of the outcome distribution.  When parameters are estimated from
\(\Dtr\), the resulting guarantee is conditional on the fitted source
predictive; it does not automatically extend to an unresolved composite
parametric null.

\paragraph{Sequential test enabled by Theorem~\ref{thm:eprocess}.}
Theorem~\ref{thm:eprocess} gives the practitioner an explicit continuously
monitored test of \(H_0^{\mathrm{pred}}\).  Starting from \(M_0=1\), after a
new target input \(X_t\) and outcome \(Y_t\) are observed, update
\begin{equation}
  M_t
  =
  M_{t-1}
  \frac{\phD{Y_t}{X_t}}{\pzeroD{Y_t}{X_t}}.
  \label{eq:sequential-update}
\end{equation}
If \(M_t\le1/\alpha\), monitoring may continue and the process is updated
again when the next target outcome becomes available.  At the first time
\(M_t>1/\alpha\), stop and reject the source predictive null.  No monitoring
horizon needs to be fixed in advance, the process may be inspected after every
observation, and the stopping decision may depend on the entire observed
history.  Inputs may also be selected adaptively.  Despite these freedoms,
if \(H_0^{\mathrm{pred}}\) is true, the probability of ever rejecting it is at
most \(\alpha\).

The formal output of this test is therefore an anytime-valid rejection of the
source predictive null.  Because every update in \eqref{eq:sequential-update}
is the prespecified likelihood ratio of \(\phDobj\) to \(\pzeroDobj\), the
rejection has a directional interpretation: the incoming target outcomes have
provided sufficient sequential evidence favoring \(\phDobj\) over
\(\pzeroDobj\).  We call this conclusion \emph{anytime-valid relative
confirmation} of the proposed correction.  The word ``relative'' emphasizes
that the conclusion compares the corrected predictive with the source
predictive; the theorem does not treat \(\phDobj\) as a null hypothesis to be
accepted.

\paragraph{Log-score representation of the evidence.}
Let
\[
  \NLPD_i(p)
  =
  -\log p(Y_i\mid X_i,\Dtr)
\]
denote the negative log-predictive density at observation \(i\).  Whenever
the two predictive log densities are finite at the observed outcome, the
one-step log evidence equals the difference in predictive log scores:
\begin{align}
  \log e_i
  &=
  \log\phD{Y_i}{X_i}
  -
  \log\pzeroD{Y_i}{X_i}
  \nonumber\\
  &=
  \NLPD_i(\pzeroDobj)
  -
  \NLPD_i(\phDobj).
  \label{eq:nlpd-gap}
\end{align}
Consequently,
\begin{equation}
  \log M_t
  =
  \sum_{i=1}^t
  \left\{
    \NLPD_i(\pzeroDobj)
    -
    \NLPD_i(\phDobj)
  \right\}.
  \label{eq:nlpd-process}
\end{equation}
Thus, on paths for which these predictive log scores are finite,
\(\log M_t\) is the cumulative predictive log-score advantage of the corrected
predictive over the source predictive on the observed target stream.  Under
the source predictive null this qualification holds almost surely whenever
the one-step log evidence is finite.  At a boundary crossing,
\[
  \log M_{\tau^*}>\log(1/\alpha),
\]
so the corrected predictive has accumulated more than \(\log(1/\alpha)\) nats
of observed log-score advantage.  This identity explains why rejection of the
source predictive null can be interpreted as relative evidence for the
prespecified correction.

\paragraph{What relative confirmation does and does not establish.}
Relative confirmation is a finite-sample, observed-data conclusion.  It says
that the target stream has accumulated enough evidence to reject the source
predictive null in the prespecified direction \(\phDobj/\pzeroDobj\).  It does
not establish that \(\phDobj\) equals the true target predictive distribution,
that the tilt \(h\) is unique or correctly specified, or that \(\phDobj\) is
close to the target distribution in an absolute sense.  The true target
predictive may be a third distribution \(q\) that differs from both
\(\pzeroDobj\) and \(\phDobj\), while \(\phDobj\) is nevertheless less wrong
than \(\pzeroDobj\) and therefore accumulates positive evidence.

The next subsection makes this population comparison precise.  Under a target
predictive \(q\), positive expected log-evidence at an input is equivalent to
\(\phDobj\) being closer to \(q\) than \(\pzeroDobj\) is in conditional KL
divergence, subject to the stated finiteness conditions.  This is a statement
of relative predictive superiority, not an absolute adequacy certificate:
even when \(\phDobj\) is closer to \(q\) than \(\pzeroDobj\), it may still be
far from \(q\).  Conversely, failure to cross the boundary does not establish
that the source predictive is correct or that the two predictives are
equivalent; it means only that the observed stream has not supplied enough
evidence for rejection at the chosen anytime-valid level.

The construction above uses the exact normalizer \(\ZhD{x}\).  Certified
upper-bound normalizers preserve confirmation validity but subtract a
predictable log-evidence penalty and no longer yield an exact predictive
likelihood ratio.  This implementation issue is treated in
\cref{sec:procedure}.  \Cref{fig:monitoring-schematic} summarizes the complete
monitoring logic.

\begin{figure}[t]
\centering
\footnotesize
\begin{tikzpicture}[
  >=Latex,
  box/.style={
    draw,rounded corners=4pt,thick,align=center,
    inner sep=5pt,minimum height=0.9cm
  },
  source/.style={
    box,fill=blue!6,draw=blue!55!black,text width=2.35cm
  },
  corr/.style={
    box,fill=orange!8,draw=orange!65!black,text width=2.05cm
  },
  model/.style={
    box,fill=green!7,draw=green!45!black,text width=2.75cm
  },
  evidence/.style={
    box,fill=violet!7,draw=violet!55!black,text width=2.55cm
  },
  decision/.style={
    box,text width=3.15cm,minimum height=1.0cm
  },
  arr/.style={->,thick}
]

\node[source] (p0) at (-4.75,2.15)
  {Source predictive\\$p_0(y\mid x, \Dtr)$};

\node[corr] (h) at (-4.75,0.95)
  {Prespecified tilt\\$h(x,y)$};

\node[model] (ph) at (-1.25,1.55)
  {Corrected predictive\\
   $p_h(y\mid x, \Dtr)$};

\node[model] (obs) at (-4.75,-0.65)
  {Incoming pair\\$(X_i,Y_i)$};

\node[evidence] (ei) at (-1.25,-0.65)
  {One-step evidence\\
   $e_i=h(X_i,Y_i)/Z_h(X_i)$};

\node[evidence] (mt) at (2.05,-0.65)
  {Running wealth\\
   $M_t=\prod_{i\le t}e_i$};

\node[decision,fill=green!8,draw=green!45!black]
  (up) at (6.60,0.95)
  {$M_t>1/\alpha$\\[3pt]
   relative confirmation};

\node[decision,fill=gray!8,draw=gray!55]
  (mid) at (6.60,-0.65)
  {$\alpha\le M_t\le1/\alpha$\\[3pt]
   continue monitoring};

\node[decision,fill=red!6,draw=red!55!black]
  (down) at (6.60,-2.25)
  {$M_t<\alpha$\\[3pt]
   relative refutation\\[2pt]
   {\scriptsize exact normalization only}};

\draw[arr] (p0.east) -- (ph.west);
\draw[arr] (h.east) -- (ph.west);

\draw[arr] (obs.east) -- (ei.west);
\draw[arr] (ph.south) -- ++(0,-0.45) -| (ei.north);

\draw[arr] (ei.east) -- (mt.west);


\draw[thick] (mt.east) -- (4.35,-0.65);

\draw[thick]
  (4.35,0.95) -- (4.35,-2.25);

\draw[arr] (4.35,0.95)  -- (up.west);
\draw[arr] (4.35,-0.65) -- (mid.west);
\draw[arr] (4.35,-2.25) -- (down.west);

\node[align=center,text width=5.4cm] at (-0.10,-2.40)
  {Under the source predictive null, exact normalization makes\\
   $e_i=p_h(Y_i\mid X_i)/p_0(Y_i\mid X_i)$ almost surely.};

\end{tikzpicture}

\caption{Schematic of sequential monitoring for a prespecified predictive
correction. The tilt transforms the source predictive into a corrected
predictive, each incoming outcome contributes one-step evidence, and the
running product is monitored continuously. The upper boundary provides
anytime-valid relative confirmation under the source predictive null.
The reciprocal lower boundary is available only with exact normalization
and is calibrated under the corrected predictive null.}
\label{fig:monitoring-schematic}
\end{figure}
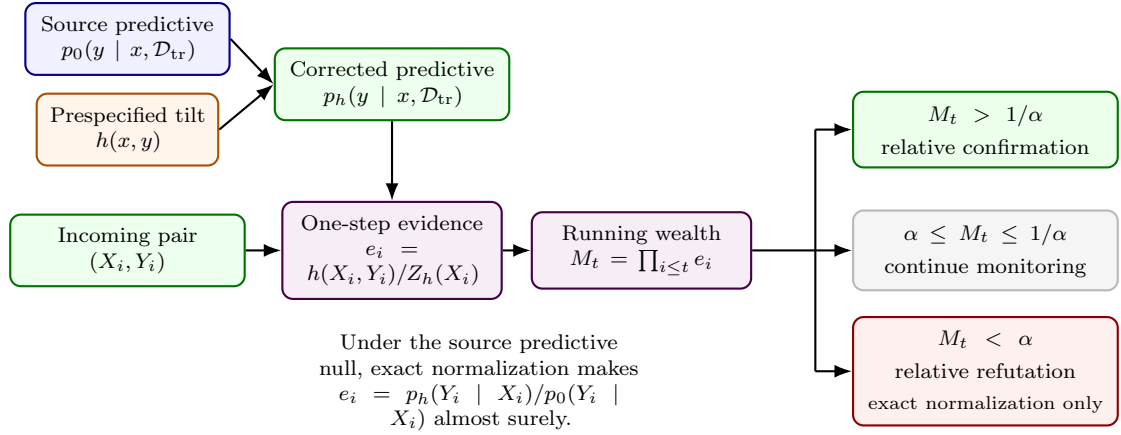

\subsection{Growth Under Alternatives}
\label{sec:growth}

The e-process guarantee controls false confirmation under the source predictive null, 
but it does not describe how evidence behaves under a target predictive distribution $q$. 
We now ask when the proposed correction accumulates evidence under $q$ and how the rate of accumulation 
depends on the inputs observed. 
The main result decomposes the log e-process into conditional expected growth under $q$ 
along the realized input sequence and a martingale fluctuation term. We first state this decomposition for general, 
possibly adaptively selected inputs; i.i.d. and stationary-ergodic limits then follow as corollaries.

For a fixed target conditional distribution $q$ and a given input-selection
mechanism, let $P_q$ denote the induced distribution of the sequential data
stream, and write $\E_q$ and $\Prob_q$ for expectation and probability under
$P_q$.  Thus the subscript $q$ specifies the outcome mechanism together with
the input process under consideration.

\bigskip
\begin{proposition}[Conditional drift decomposition]
\label{prop:drift}
Suppose that, for each $i$, conditional on \(\cG_i\), the outcome satisfies
\[
  Y_i\mid\cG_i\sim q(\cdot\mid X_i)
\]
for a fixed target conditional distribution $q$. Assume that
\[
  \E_q \! \left[
    \left|
      \log\frac{h(X_i,Y_i)}{\ZhD{X_i}}
    \right|
  \right]
  <\infty
\]
for every $i$. Define a measurable version of the per-input drift by
\begin{equation}
  \Gamma_h(x)
  \defeq
  \int q(y\mid x)
  \log\frac{h(x,y)}{\ZhD{x}}\,dy .
  \label{eq:gamma-x}
\end{equation}
wherever the integral is finite. The preceding integrability condition ensures
that $\Gamma_h(X_i)$ is finite almost surely for every $i$.
Whenever both KL divergences are finite, the drift can equivalently be written as
\begin{equation}
  \Gamma_h(x)
  =
  \KL\!\left(
    q(\cdot\mid x)\,\middle\|\,\pzeroD{\cdot}{x}
  \right)
  -
  \KL\!\left(
    q(\cdot\mid x)\,\middle\|\,\phD{\cdot}{x}
  \right).
  \label{eq:gamma-kl}
\end{equation}

Then, under the target process induced by $q$,
\[
  \E_q[\log e_i\mid\cG_i]
  =
  \int q(y\mid X_i)
  \log\frac{h(X_i,y)}{\ZhD{X_i}}\,dy
  =
  \Gamma_h(X_i)
  \qquad\text{a.s.},
\]
and $\E_q|\Gamma_h(X_i)|<\infty$. Consequently,
\[
  N_t
  \defeq
  \log M_t-\sum_{i=1}^t\Gamma_h(X_i),
  \qquad N_0=0,
\]
is a martingale with respect to $(\cF_t)_{t\ge0}$ under $P_q$.

If, in addition, there exists a finite constant $v$ such that
\[
  \E_q \!\left[
    \bigl\{\log e_i-\Gamma_h(X_i)\bigr\}^2
    \,\middle|\,
    \cG_i
  \right]
  \le v
  \qquad\text{a.s. for every }i,
\]
then
\[
  \frac{N_t}{t}\longrightarrow0
  \qquad\text{a.s.}
\]
Therefore, on the event
\[
  \left\{
    \liminf_{t\to\infty}
    \frac{1}{t}\sum_{i=1}^t\Gamma_h(X_i)>0
  \right\},
\]
we have $\log M_t\to+\infty$ and hence $\tau^*<\infty$ almost surely.
\end{proposition}

\begin{proof}[Proof sketch]
Conditioning on $\cG_i$ and using
$Y_i\mid\cG_i\sim q(\cdot\mid X_i)$ gives
\[
  \E_q[\log e_i\mid\cG_i]=\Gamma_h(X_i).
\]
Hence the centered increments
\[
  \log e_i-\Gamma_h(X_i)
\]
form a martingale difference sequence under $P_q$, so $N_t$ is an
$(\cF_t)$-martingale under $P_q$.  The conditional second-moment bound implies, by a
martingale strong law, that $N_t/t\to0$ almost surely.  Therefore,
\[
  \frac{\log M_t}{t}
  =
  \frac1t\sum_{i=1}^t\Gamma_h(X_i)+\frac{N_t}{t},
\]
and a positive lower limit of the average drift forces $\log M_t\to+\infty$, so the
confirmation boundary is crossed in finite time.  See
\cref{app:proof-drift} for details.
\end{proof}

Proposition~\ref{prop:drift} has three main implications.  
\begin{itemize}
\item[a.]
First, the sign of $\Gamma_h(x)$ measures relative predictive merit at input $x$.  Whenever the two KL divergences in \eqref{eq:gamma-kl} are finite, $\Gamma_h(x)>0$ exactly when the corrected predictive $\phD{\cdot}{x}$ is closer to the true target conditional distribution $q(\cdot\mid x)$ in KL divergence than the source predictive $\pzeroD{\cdot}{x}$ is.  The correction need not coincide with the true target distribution: an incorrect magnitude, or even an imperfect structural form, can have positive expected log-growth if it predicts better than the source model.  Conversely, scientific plausibility alone does not ensure positive drift.  Confirmation therefore concerns the correction's predictive advantage relative to the source, not exact estimation or identification of the shift.
\item[b.]
Second, the decomposition
\[
  \log M_t=\sum_{i=1}^t\Gamma_h(X_i)+N_t
\]
separates systematic evidence growth from random fluctuation.  The first term is the cumulative conditional expected log-score advantage of the correction along the inputs actually observed; $N_t$ records the deviations of the realized log scores from those conditional expectations.  Under the conditional second-moment condition, $N_t/t\to0$ almost surely.  Consequently, persistent positive average drift implies eventual confirmation with probability one, whereas the proposition itself does not provide a finite-horizon power
function such as $\Prob_q(\tau^*\le t)$; a conservative finite-horizon bound
under i.i.d.\ sampling is given in \cref{cor:finite-horizon}.
\item[c.]
Third, the input sequence affects the rate of evidence accumulation through $\Gamma_h(X_i)$.  Inputs with large positive drift are more informative for comparing $\phDobj$ with $\pzeroDobj$, while inputs with drift near zero contribute little expected log evidence.  The null guarantee remains valid for arbitrary adaptive input selection because it conditions on the realized $X_i$.  Thus, when informative inputs can be identified from scientific knowledge or a prespecified or predictable design criterion, adaptive experimental design may accelerate confirmation without changing the source-null error guarantee.  The design affects the growth rate under the target alternative, not the validity of the e-process under the null.
\end{itemize}

\bigskip
\begin{corollary}[Asymptotic growth under i.i.d.\ or stationary-ergodic sampling]
\label{cor:ergodic}
Suppose that either

\begin{itemize}
\item[(i)] the pairs $(X_i,Y_i)_{i\ge1}$ are i.i.d.\ with joint distribution
\[
  q^X(dx)\,q(dy\mid x),
\]
or
\item[(ii)] the pair process $(X_i,Y_i)_{i\ge1}$ is stationary and ergodic with
one-step distribution
\[
  q^X(dx)\,q(dy\mid x).
\]
\end{itemize}

If
\[
  \E_{q^X q^{Y\mid X}}\!\left[
    \left|
      \log\frac{h(X,Y)}{\ZhD{X}}
    \right|
  \right]
  <\infty,
\]
then
\begin{equation}
  \frac{1}{t}\log M_t
  \longrightarrow
  \overline\Gamma(q;h)
  \defeq
  \E_{X\sim q^X}[\Gamma_h(X)]
  =
  \E_{(X,Y)\sim q^X q^{Y\mid X}}\!\left[
    \log\frac{h(X,Y)}{\ZhD{X}}
  \right]
  \qquad\text{a.s.}
  \label{eq:ergodic-growth}
\end{equation}
If, in addition, the two expected predictive log losses below are finite, then
this limit has the equivalent log-score representation
\begin{equation}
  \overline\Gamma(q;h)
  =
  \E_{q^X q^{Y\mid X}}\!\left[\NLPD(\pzeroDobj)\right]
  -
  \E_{q^X q^{Y\mid X}}\!\left[\NLPD(\phDobj)\right].
  \label{eq:ergodic-logloss}
\end{equation}
Consequently, if $\overline\Gamma(q;h)>0$, then
$\log M_t\to+\infty$ and $\tau^*<\infty$ almost surely.  If
$\overline\Gamma(q;h)<0$, then $\log M_t\to-\infty$ and $M_t\to0$ almost surely.
\end{corollary}

\begin{proof}[Proof sketch]
Under either assumption, the sequence
\[
  \log e_i
  =
  \log\frac{h(X_i,Y_i)}{\ZhD{X_i}}
\]
is integrable and is respectively i.i.d.\ or stationary and ergodic.
The ordinary strong law or Birkhoff's ergodic theorem therefore gives
\[
  \frac{1}{t}\log M_t
  =
  \frac{1}{t}\sum_{i=1}^t\log e_i
  \longrightarrow
  \E_{q^X q^{Y\mid X}}[\log e_1]
  =
  \E_{q^X}[\Gamma_h(X)]
  \qquad\text{a.s.}
\]
The conclusions for positive and negative limits follow immediately.
See \cref{app:proof-ergodic} for details.
\end{proof}

The corollary replaces the input-dependent cumulative drift in
Proposition~\ref{prop:drift} by a single deterministic long-run growth rate.
Under i.i.d.\ or stationary-ergodic sampling,
\[
  \log M_t
  =
  t\,\overline\Gamma(q;h)+o(t)
  \qquad\text{a.s.}
\]
Thus $\overline\Gamma(q;h)$ is the asymptotic number of nats of evidence gained
per observation.  It is positive exactly when the corrected predictive
$\phDobj$ has smaller expected negative log-predictive density than the
source predictive $\pzeroDobj$ under the joint target distribution
$q^X(dx)q(dy\mid x)$.

When $\overline\Gamma(q;h)>0$, the e-process grows exponentially at rate
$\overline\Gamma(q;h)$:
\[
  M_t
  =
  \exp\{t\overline\Gamma(q;h)+o(t)\}.
\]
The correction is therefore eventually confirmed almost surely, and the
first-order crossing-time approximation is
\[
  \tau^*
  \approx
  \frac{\log(1/\alpha)}{\overline\Gamma(q;h)}.
\]
When $\overline\Gamma(q;h)<0$, the source predictive has the smaller expected
log loss and the evidence process decays exponentially.  When
$\overline\Gamma(q;h)=0$, neither predictive has a long-run expected log-score
advantage, and the corollary alone does not determine whether a finite
boundary crossing occurs.

Proposition~\ref{prop:drift} and Corollary~\ref{cor:ergodic} serve
complementary purposes.  Proposition~\ref{prop:drift} applies to general,
possibly adaptively selected inputs and describes growth through the
path-dependent average
\[
  \frac{1}{t}\sum_{i=1}^t\Gamma_h(X_i).
\]
Corollary~\ref{cor:ergodic} applies when the sampling process has a stable
long-run distribution and reduces this quantity to the population average
\[
  \overline\Gamma(q;h)=\E_{q^X}[\Gamma_h(X)].
\]
It therefore provides a simple summary of the correction's long-run
predictive advantage and connects the sequential e-process directly to
standard expected log-loss comparison.

\bigskip
\begin{corollary}[A finite-horizon crossing bound]
\label{cor:finite-horizon}
Suppose that the log e-values $Z_i=\log e_i$ are i.i.d.\ under $P_q$, with
\[
  \E_q[Z_i]=\overline\Gamma(q;h)>0,
  \qquad
  \Var_q(Z_i)\le v<\infty.
\]
Let $b=\log(1/\alpha)$.  For every integer $t$ satisfying
$t\overline\Gamma(q;h)>b$,
\begin{equation}
  \Prob_q(\tau^*>t)
  \le
  \frac{t v}{
    t v+\{t\overline\Gamma(q;h)-b\}^2
  }.
  \label{eq:finite-horizon-power}
\end{equation}
If, in addition, the centered increments $Z_i-\overline\Gamma(q;h)$ are
sub-Gaussian with variance proxy $s^2$, that is,
\[
  \E_q\exp\left[\theta\{Z_i-\overline\Gamma(q;h)\}\right]
  \le
  \exp\left(\frac{\theta^2s^2}{2}\right)
  \qquad\text{for every }\theta\in\R,
\]
then the same event admits the exponential bound
\begin{equation}
  \Prob_q(\tau^*>t)
  \le
  \exp\left[
    -\frac{\{t\overline\Gamma(q;h)-b\}^2}{2ts^2}
  \right].
  \label{eq:finite-horizon-subgaussian}
\end{equation}
If instead the increments are bounded, with
$|Z_i-\overline\Gamma(q;h)|\le R$ almost surely, then
\begin{equation}
  \Prob_q(\tau^*>t)
  \le
  \exp\left[
    -\frac{\{t\overline\Gamma(q;h)-b\}^2}
          {2tv+\tfrac23R\{t\overline\Gamma(q;h)-b\}}
  \right].
  \label{eq:finite-horizon-bernstein}
\end{equation}
\end{corollary}

\begin{proof}[Proof sketch]
The event $\{\tau^*>t\}$ implies $\log M_t\le b$.  All three bounds follow by
applying a lower-tail inequality to
$\log M_t-t\overline\Gamma(q;h)=\sum_{i=1}^t\{Z_i-\overline\Gamma(q;h)\}$
at the deviation level $t\overline\Gamma(q;h)-b>0$: Cantelli's one-sided variance
inequality gives \eqref{eq:finite-horizon-power}, using that
$\sigma^2\mapsto\sigma^2/(\sigma^2+\lambda^2)$ is increasing, so that the
variance may be replaced by the upper bound $tv$; the sub-Gaussian Chernoff
bound gives \eqref{eq:finite-horizon-subgaussian}; and Bernstein's inequality
gives \eqref{eq:finite-horizon-bernstein}.  See
\cref{app:proof-finite-horizon} for details.
\end{proof}

The bounds make the crossing-time heuristic $\tau^*\approx b/\overline\Gamma(q;h)$
operational: once the expected accumulated log evidence exceeds the boundary,
the probability of not yet crossing is explicitly controlled.  The two differ
sharply in how fast that control improves.  The variance-only bound
\eqref{eq:finite-horizon-power} decays only at the polynomial rate
$v/\{t\overline\Gamma(q;h)^2\}$ and is therefore very conservative at moderate
horizons, whereas \eqref{eq:finite-horizon-subgaussian} decays exponentially in
$t$.  For the Gaussian mean tilt of \cref{sec:mean-shift} with bounded $g$ the
increments are sub-Gaussian, so the exponential bound also applies.  Its
numerical sharpness depends entirely on the certified variance proxy;
\cref{sec:exp-mean-shift} evaluates both the variance-only bound and a
deliberately conservative certified sub-Gaussian proxy.

The Bernstein form \eqref{eq:finite-horizon-bernstein} is useful when a
deterministic bound on the centered increments and a variance bound are both
available, especially when the increment distribution is strongly skewed.  A
sub-Gaussian proxy obtained only from a worst-case range can be much looser
because it discards the variance information, whereas
\eqref{eq:finite-horizon-bernstein} uses the variance and the range together.
The two forms can cross: \eqref{eq:finite-horizon-power} may be sharper at short
horizons, where the linear term $\tfrac23R\{t\overline\Gamma(q;h)-b\}$ dominates the
Bernstein denominator, while \eqref{eq:finite-horizon-bernstein} can become
sharper at long horizons as the quadratic numerator grows.  The Bernstein bound
requires valid variance and range bounds; these quantities should not be
estimated from the same monitored outcomes and then treated as prospective
certificates.  A step-by-step interpretation of the crossing event, the
sample-size heuristic, and all three finite-horizon bounds is provided in
\cref{app:finite-horizon-guide}.

\bigskip
\begin{corollary}[Correctly specified predictive correction]
\label{cor:correct}
Suppose that the target conditional distribution is exactly the corrected
predictive:
\[
  q(\cdot\mid X_i)=\phD{\cdot}{X_i}
  \qquad\text{a.s. for every }i.
\]
Under the integrability hypothesis of \cref{prop:drift},
\begin{equation}
  \Gamma_h(X_i)
  =
  \KL\!\left(
    \phD{\cdot}{X_i}
    \,\middle\|\,
    \pzeroD{\cdot}{X_i}
  \right)
  \ge 0
  \qquad\text{a.s.}
  \label{eq:correct-drift}
\end{equation}
If, in addition, the bounded conditional second-moment condition of
\cref{prop:drift} holds, then for arbitrary, possibly adaptively selected
inputs,
\[
  \tau^*<\infty
  \qquad\text{a.s. on the event}\qquad
  \left\{
    \liminf_{t\to\infty}
    \frac{1}{t}\sum_{i=1}^t
    \KL\!\left(
      \phD{\cdot}{X_i}
      \,\middle\|\,
      \pzeroD{\cdot}{X_i}
    \right)
    >0
  \right\}.
\]

Under either the i.i.d.\ or stationary-ergodic sampling regime of
\cref{cor:ergodic}, suppose instead that its integrability condition holds.
Then
\begin{equation}
  \frac{1}{t}\log M_t
  \longrightarrow
  \gamma_h
  \defeq
  \E_{X\sim q^X}\!\left[
    \KL\!\left(
      \phD{\cdot}{X}
      \,\middle\|\,
      \pzeroD{\cdot}{X}
    \right)
  \right]
  \ge 0
  \qquad\text{a.s.}
  \label{eq:correct-growth-rate}
\end{equation}
In particular, if $\gamma_h>0$, then $\log M_t\to+\infty$ and
$\tau^*<\infty$ almost surely.  Moreover, $\gamma_h=0$ if and only if
\[
  \phD{\cdot}{X}
  =
  \pzeroD{\cdot}{X}
  \qquad
  \text{for $q^X$-almost every $X$}.
\]
\end{corollary}

\begin{proof}[Proof sketch]
Setting $q(\cdot\mid x)=\phD{\cdot}{x}$ in \eqref{eq:gamma-kl} gives
\[
  \Gamma_h(x)
  =
  \KL\!\left(
    \phD{\cdot}{x}
    \,\middle\|\,
    \pzeroD{\cdot}{x}
  \right),
\]
because
\[
  \KL\!\left(
    \phD{\cdot}{x}
    \,\middle\|\,
    \phD{\cdot}{x}
  \right)
  =0.
\]
The arbitrary-input conclusion follows from \cref{prop:drift}, and the
asymptotic growth statement follows from \cref{cor:ergodic}.  See
\cref{app:proof-correct} for details.
\end{proof}

Under correct specification, let $P_h$ denote the data-stream distribution
induced by the corrected predictive and the given input mechanism.  The
general conditional drift then reduces to an information divergence:
\[
  \E_{P_h}[\log e_i\mid\cG_i]
  =
  \KL\!\left(
    \phD{\cdot}{X_i}
    \,\middle\|\,
    \pzeroD{\cdot}{X_i}
  \right).
\]
Thus no input has negative expected
log-growth.  An input contributes zero expected evidence when the corrected
and source predictives coincide there, and positive expected evidence when
they differ, subject to the stated integrability conditions.

For arbitrary, possibly adaptive, inputs, eventual confirmation requires the
average KL separation along the realized input sequence to remain positive.
Correct specification alone is therefore not enough if the sampling mechanism
visits only regions where the two predictives are indistinguishable.  This also
gives the result an experimental-design interpretation: inputs with larger KL
separation are more informative for confirming the correction.

Under i.i.d.\ or stationary-ergodic sampling, the pathwise average reduces to
$\gamma_h$, the population-average KL separation.  Hence $\gamma_h$ is the
asymptotic number of nats of evidence gained per observation, and
$\gamma_h>0$ implies exponential evidence growth and eventual confirmation
almost surely.  This is the clean benchmark case: when the prespecified
correction is the true target predictive, its evidence rate is exactly the KL
information separating it from the source predictive.

\subsection{False-Confirmation Control under Target Misspecification}
\label{sec:misspecification-control}

Theorem~\ref{thm:eprocess} controls the probability of ever confirming the
correction when the source predictive null is true.  In deployment, however,
the true target predictive distribution $q$ may be neither the source
predictive $\pzeroDobj$ nor the proposed corrected predictive $\phDobj$.
The source predictive null remains the original null hypothesis; a target
distribution $q\ne\pzeroDobj$ is not reclassified as part of that null.
Instead, we ask a robustness question:

\begin{quote}
For which misspecified target predictive distributions $q$ does the original
stopping rule still control the probability of ever crossing the confirmation
boundary by $\alpha$?
\end{quote}

Answering this question identifies a protected class of target
misspecifications for the original directional bet.  It does not enlarge the
source predictive null.  Rather, it clarifies when the same maximal crossing
bound persists despite misspecification and when additional protection must
be built into the e-process.  The following condition is exactly the condition
under which each one-step factor remains a conditional e-value under the
misspecified target process.

\bigskip
\begin{proposition}[Persistence of false-confirmation control under target misspecification]
\label{prop:misspecification-control}
For each $i$, let $q_i$ be a $\cG_i$-measurable target conditional
distribution and suppose that
\[
  Y_i\mid\cG_i\sim q_i(\cdot\mid X_i).
\]
The sequence $(q_i)$ may vary predictably with time, the past, and the current
input.  If
\begin{equation}
  \E_{Y\sim q_i(\cdot\mid X_i)}
  \bigl[h(X_i,Y)\bigr]
  \le
  \E_{Y\sim\pzeroD{\cdot}{X_i}}
  \bigl[h(X_i,Y)\bigr]
  =
  \ZhD{X_i}
  \qquad\text{a.s. for every }i,
  \label{eq:moment-condition}
\end{equation}
then $(M_t)_{t\ge0}$ is a nonnegative supermartingale under the data-stream
distribution $P_{(q_i)}$ induced by $(q_i)$ and the given input mechanism.
Consequently,
\[
  \Prob_{(q_i)}\!\left\{
    \sup_{t\ge0}M_t>\frac{1}{\alpha}
  \right\}
  \le\alpha.
\]
In particular, if every conditional distribution in a class
$\mathcal Q_0$ satisfies
\[
  \E_{Y\sim r(\cdot\mid x)}[h(x,Y)]
  \le \ZhD{x}
  \qquad\text{for every }x,
\]
then
\[
  \sup_{(q_i):\,q_i\in\mathcal Q_0}
  \Prob_{(q_i)}\!\left\{
    \sup_{t\ge0}M_t>\frac{1}{\alpha}
  \right\}
  \le\alpha,
\]
where the supremum is over all predictable selections from $\mathcal Q_0$, and
the bound holds for any input process, including an adaptively selected one.
\end{proposition}

\begin{proof}[Proof sketch]
Because
\[
  e_i=\frac{h(X_i,Y_i)}{\ZhD{X_i}},
\]
condition \eqref{eq:moment-condition} is equivalent to
\[
  \E[e_i\mid\cG_i]\le1.
\]
Iterated conditioning therefore makes $(M_t)$ a nonnegative supermartingale
under the induced process, and Ville's inequality gives the stated crossing
bound.  See
\cref{app:proof-misspecification-control} for details.
\end{proof}

\paragraph{What this result establishes.}
Whenever \eqref{eq:moment-condition} holds along the inputs visited by the
process, the same time-uniform crossing bound continues to hold under the
misspecified target process:
\[
  \Prob_{(q_i)}\!\left\{\sup_{t\ge0}M_t>1/\alpha\right\}\le\alpha.
\]
Thus the original directional test is robust to a correction-dependent class
of target misspecifications, even though those targets are not part of the
source predictive null.  The protected class is characterized below.

This distinction matters because a departure from $\pzeroDobj$ need not favor
the proposed correction.  Some departures move outcomes in the opposite
direction or leave the moment targeted by $h$ unchanged.  The proposition
identifies a correction-specific region in which such departures still cannot
inflate the probability of false confirmation beyond $\alpha$.

\subsubsection{A Correction-Dependent Protected Half-Space}
\label{sec:protected-halfspace}
At a fixed input $x$, define
\begin{equation}
  \mathcal H_h(x)
  \defeq
  \left\{
    q(\cdot\mid x):
    \int h(x,y)\,q(dy\mid x)\le\ZhD{x}
  \right\}.
  \label{eq:halfspace}
\end{equation}
The mapping
\[
  q(\cdot\mid x)
  \longmapsto
  \int h(x,y)\,q(dy\mid x)
\]
is linear in the target conditional distribution.  Hence $\mathcal H_h(x)$ is
the intersection of the set of conditional distributions with a linear
half-space.  The source predictive lies on its boundary because
\[
  \int h(x,y)\,\pzeroD{y}{x}\,dy=\ZhD{x}.
\]
Every target conditional distribution on the protected side makes the
one-step evidence factor have conditional mean at most one.  Consequently, if
a predictable target sequence satisfies
$q_i(\cdot\mid X_i)\in\mathcal H_h(X_i)$ almost surely at every monitored step,
the original wealth process remains a nonnegative supermartingale and retains
the same time-uniform crossing bound.  The orientation of this protected
region is determined entirely by the prespecified correction $h$; it is a
robustness region for the directional bet, not an enlargement of the original
null hypothesis.

The same half-space has a direct connection to the growth analysis in
\cref{sec:growth}.  For any $q(\cdot\mid x)\in\mathcal H_h(x)$, Jensen's
inequality gives
\[
  \E_q[\log e_i\mid X_i=x]
  \le
  \log\E_q[e_i\mid X_i=x]
  \le0,
\]
whenever the logarithmic expectation is well defined.  Thus no target on the
protected side can have positive conditional expected log evidence in favor
of the correction at that input.

The corrected predictive lies strictly outside this protected half-space
whenever the correction is nontrivial at $x$, meaning that
$h(x,Y)/\ZhD{x}$ is not equal to one $\pzeroD{\cdot}{x}$-almost surely.  Indeed,
\begin{align}
  \E_{Y\sim\phD{\cdot}{x}}
  \left[
    \frac{h(x,Y)}{\ZhD{x}}
  \right]
  &=
  \E_{Y\sim\pzeroD{\cdot}{x}}
  \left[
    \left\{
      \frac{h(x,Y)}{\ZhD{x}}
    \right\}^{2}
  \right]
  \notag\\
  &=
  1+
  \Var_{Y\sim\pzeroD{\cdot}{x}}
  \left(
    \frac{h(x,Y)}{\ZhD{x}}
  \right)
  \in(1,\infty],
  \label{eq:fails-at-ph}
\end{align}
where the second moment is interpreted in the extended sense.  This is exactly
what should happen: under the corrected predictive represented by the
alternative, the one-step evidence factor is expected to grow rather than to
retain the supermartingale property used for false-confirmation control.

\paragraph{Why generic closeness to the source is insufficient.}
The protected half-space is directional; it is not a KL, total-variation, or
Hellinger neighborhood around the source predictive.  For every nontrivial
tilt, there are target distributions outside $\mathcal H_h(x)$ that are
arbitrarily close to $\pzeroD{\cdot}{x}$ in total variation, Hellinger
distance, and $\KL\!\left(q(\cdot\mid x)\,\middle\|\,\pzeroD{\cdot}{x}\right)$.
Thus being close to the source in a generic distributional metric does not by
itself preserve the level-$\alpha$ crossing guarantee.  What matters for the
original bet is the correction-specific moment comparison in
\eqref{eq:moment-condition}.  A construction establishing this claim is given
in \cref{app:halfspace-details}.

\subsubsection{What Failure of the Moment Condition Means}
\label{sec:moment-condition-failure}
If \eqref{eq:moment-condition} fails at inputs visited by the process, the
one-step factor is no longer guaranteed to be a conditional e-value under
$q$, and the original supermartingale proof is unavailable.  This failure does
not by itself imply that
\[
  \Prob_q\!\left\{
    \sup_{t\ge0}M_t>\frac{1}{\alpha}
  \right\}
  >
  \alpha.
\]
The moment condition is sufficient for maximal crossing control and exact for
the one-step conditional e-value property, but failure of that sufficient
condition is not a converse false-confirmation result.  The crossing
probability may still be at most $\alpha$ for other, distribution-specific
reasons; it is simply no longer controlled by \cref{prop:misspecification-control}.

Negative long-run log-drift does not restore the missing anytime-valid
guarantee.  Under appropriate ergodic conditions, negative drift implies
$\log M_t\to-\infty$ almost surely, but the error criterion concerns the
probability of at least one boundary crossing over the entire path.  A process
that eventually decays may still cross early because of high-variance
increments or a single heavy-tailed jump.  Thus asymptotic decay and control
of the maximal process are distinct properties.  The cross-family
calculations in \cref{sec:cross-family} and experiments in
\cref{sec:exp-cross-family} illustrate both mechanisms.

The main conclusion is therefore directional.  The unmodified e-process has
its original level-$\alpha$ guarantee under the source predictive null and
retains the same time-uniform crossing bound for any predictable target
sequence that stays in the protected half-space at the inputs actually
visited.  It does not automatically protect a generic neighborhood of the
source or an arbitrary user-chosen class of plausible target distributions.
Broader uniform protection over a user-specified class would require
redesigning the e-process for that class and is beyond the scope of the present
paper.

\subsection{Anytime-Valid Confirmation, Refutation, and Overshoot}
\label{sec:refutation}

The preceding results use an upper boundary to reject the source predictive
null in the direction of the proposed correction.  A practitioner may also
want to stop in the opposite direction when the incoming target outcomes
provide sufficient evidence against the corrected predictive itself.  This
section asks:

\begin{quote}
Can the same monitored wealth process support both anytime-valid relative
confirmation and anytime-valid refutation of the proposed correction?
\end{quote}

\subsubsection{Anytime-Valid Two-Boundary Decisions}
\label{sec:two-boundary}

The answer is yes for the original, exactly normalized likelihood-ratio
process $M_t$ from \cref{thm:eprocess}.  It relies on the reciprocal likelihood ratio and therefore does not extend to a process formed using a conservative upper bound on the normalizer.

Suppose that
$h(x,y)>0$ for $\pzeroD{\cdot}{x}$-almost every $y$, so that
$\pzeroDobj$ and $\phDobj$ are mutually absolutely continuous.  Define the
corrected predictive null
\[
  H_h^{\mathrm{pred}}:\quad
  Y_i\mid\cG_i\sim \phD{\cdot}{X_i}.
\]
Let $\mathcal P_h^{\mathrm{pred}}$ denote the corresponding class of all
data-stream distributions, again allowing any admissible input mechanism.
The reciprocal e-values are
\[
  e_i^{\downarrow}
  =\frac{\pzeroD{Y_i}{X_i}}{\phD{Y_i}{X_i}}
  =\frac{\ZhD{X_i}}{h(X_i,Y_i)},
  \qquad
  M_t^{\downarrow}
  =\prod_{i=1}^t e_i^{\downarrow}
  =\frac{1}{M_t}.
\]
By the same argument as \cref{thm:eprocess}, with the roles of
$\pzeroDobj$ and $\phDobj$ exchanged, $M_t^{\downarrow}$ is a nonnegative
martingale under $H_h^{\mathrm{pred}}$, so
\[
  \sup_{P\in\mathcal P_h^{\mathrm{pred}}}
  \Prob_P\!\left(
    \inf_{t\ge0}M_t<\alpha
  \right)
  =
  \sup_{P\in\mathcal P_h^{\mathrm{pred}}}
  \Prob_P\!\left(
    \sup_{t\ge0}M_t^{\downarrow}>\frac{1}{\alpha}
  \right)
  \le\alpha.
\]
\paragraph{The two-boundary sequential decision.}
Monitoring the single wealth process $M_t$ gives two anytime-valid rejection
rules for the simple predictive pair:
\[
  M_t>1/\alpha
  \quad\Longrightarrow\quad
  \text{reject }H_0^{\mathrm{pred}}
  \text{ and relatively confirm }\phDobj\text{ over }\pzeroDobj,
\]
and
\[
  M_t<\alpha
  \quad\Longrightarrow\quad
  \text{reject }H_h^{\mathrm{pred}}
  \text{ and refute the proposed corrected predictive in favor of }\pzeroDobj.
\]
The probability of ever making the upper rejection is at most $\alpha$ when
$H_0^{\mathrm{pred}}$ is true, and the probability of ever making the lower
rejection is at most $\alpha$ when $H_h^{\mathrm{pred}}$ is true.  Neither
boundary proves that the predictive distribution favored by that boundary is
the true target predictive: the upper boundary rejects $\pzeroDobj$, whereas
the lower boundary rejects $\phDobj$.  If neither boundary is crossed, the
procedure remains inconclusive rather than accepting either model.

This is Wald's two-boundary sequential probability ratio test
\citep{WaldA1945aoms}.  The conditional formulation allows arbitrary adaptive
input selection provided that the same input mechanism is used under the two
predictive hypotheses.  Each error guarantee holds only under its
corresponding simple null; behavior under other target distributions is
characterized in \cref{sec:misspecification-control}.

The asymmetry of that qualification deserves emphasis, because the lower
boundary is the more fragile of the two in practice.  The upper boundary is
protected under the source predictive null and, by
\cref{prop:misspecification-control}, over the whole protected half-space
$\mathcal H_h(x)$.  The lower-boundary guarantee established here is calibrated
under $H_h^{\mathrm{pred}}$, that is, when the target predictive is exactly the
proposed corrected predictive.  We do not develop an analogous
misspecification-robustness class for the reciprocal process.  A target that
matches $\phDobj$ in the feature the correction acts on but differs from it in
some other respect is outside $\mathcal P_h^{\mathrm{pred}}$, so the guarantee
proved here does not apply.  A practitioner who wants to retire a correction
should therefore treat a lower crossing as rejection of the entire corrected
predictive relative to the source, and not as evidence that the correction
magnitude alone was wrong in the direction it was designed to test.

\subsubsection{Why the Actual Null Crossing Probability Can Be Below $\alpha$}
\label{sec:overshoot}
The boundary $1/\alpha$ gives a valid upper bound on the probability of ever
falsely confirming the correction, but the actual source-null crossing
probability is generally smaller than $\alpha$.  This does not change the
decision rule; it explains its conservativeness.  At the crossing time, the
wealth usually jumps beyond $1/\alpha$ rather than landing exactly on it.  The
following result quantifies this \emph{overshoot} and, for the exact predictive
likelihood ratio, separates its contribution from the probability of eventual
crossing under the corrected predictive.

\begin{proposition}[Overshoot identity]
\label{prop:overshoot}
Let $P_0$ denote the data-stream distribution under
$H_0^{\mathrm{pred}}$, let $(M_t)$ be a nonnegative $P_0$-martingale with
$M_0=1$, and define
\[
  \tau^*=\inf\{t:M_t>1/\alpha\}.
\]
If $P_0(\tau^*<\infty)>0$, then
\begin{equation}
  \E_{P_0}
  \left[
    M_{\tau^*}\mathbf 1\{\tau^*<\infty\}
  \right]
  \le1,
  \qquad
  P_0(\tau^*<\infty)
  \le
  \frac{1}{
    \E_{P_0}[M_{\tau^*}\mid\tau^*<\infty]
  }
  <\alpha.
  \label{eq:overshoot-ineq}
\end{equation}

For the exactly normalized predictive-correction process
\[
  M_t
  =
  \prod_{i=1}^t
  \frac{\phD{Y_i}{X_i}}{\pzeroD{Y_i}{X_i}},
\]
let $P_h$ denote the data-stream distribution under
$H_h^{\mathrm{pred}}$ and suppose that $P_h$ and $P_0$ use the same,
possibly adaptive, input mechanism.  Then $P_h\ll P_0$ on every $\cF_t$,
$M_t$ is their likelihood ratio. This one-sided absolute continuity follows
from \cref{lem:ratio} and does not require the strict-positivity assumption used
for reciprocal refutation. Moreover,
\begin{equation}
  \E_{P_0}
  \left[
    M_{\tau^*}\mathbf 1\{\tau^*<\infty\}
  \right]
  =
  P_h(\tau^*<\infty).
  \label{eq:overshoot-exact}
\end{equation}
Consequently, whenever $P_0(\tau^*<\infty)>0$,
\begin{equation}
  P_0(\tau^*<\infty)
  =
  \frac{
    P_h(\tau^*<\infty)
  }{
    \E_{P_0}[M_{\tau^*}\mid\tau^*<\infty]
  }.
  \label{eq:overshoot-ratio}
\end{equation}
In particular, if $\log M_t\to+\infty$ $P_h$-almost surely---as under the
positive-drift condition of \cref{cor:correct}---then
$P_h(\tau^*<\infty)=1$, the null crossing probability is positive, and
\begin{equation}
  P_0(\tau^*<\infty)
  =
  \frac{1}{
    \E_{P_0}[M_{\tau^*}\mid\tau^*<\infty]
  }.
  \label{eq:overshoot-equality}
\end{equation}
\end{proposition}

\begin{proof}[Proof sketch]
Optional stopping for the stopped nonnegative martingale followed by Fatou's
lemma gives \eqref{eq:overshoot-ineq}.  For the exact identity, use
$M_t=dP_h|_{\cF_t}/dP_0|_{\cF_t}$ on each event $\{\tau^*=t\}$ and sum over
$t$.  See \cref{app:proof-overshoot} for details.
\end{proof}

\paragraph{Interpretation of the overshoot identity.}
The first inequality shows strict conservativeness whenever the null crossing
probability is positive, because $M_{\tau^*}>1/\alpha$ on the crossing event.
For the predictive likelihood-ratio process, \eqref{eq:overshoot-ratio} shows
that two quantities determine the source-null crossing probability: the mean
wealth at crossing under $P_0$ and the probability that the upper boundary is
ever reached under $P_h$.  If positive drift under $P_h$ makes eventual
crossing certain, then \eqref{eq:overshoot-equality} isolates the overshoot
effect exactly.  This is an accounting identity, not an additional testing
claim.  The corresponding empirical check is whether
$\Prob(\mathrm{cross})\,\E[M_{\tau^*}\mid\mathrm{cross}]$ is close to one;
\cref{sec:exp-overshoot} reports this product for the Gaussian null
experiment.

\subsection{Practical Monitoring Procedure}
\label{sec:procedure}

The preceding results yield two monitoring modes that should be selected before observing the target outcomes:
\begin{enumerate}
\item \emph{Exact relative confirmation}: use the exact normalizer, preserve the likelihood-ratio and cumulative log-score interpretations, and, when desired, monitor the reciprocal lower boundary to refute the corrected predictive.
\item \emph{Conservative source-null confirmation}: use a certified upper bound on the normalizer when exact normalization is unavailable, retaining anytime-valid rejection of the source predictive null at a predictable cost in log evidence.
\end{enumerate}

The core construction uses
\[
  \ZhD{x}=\E_{Y\sim\pzeroD{\cdot}{x}}[h(x,Y)].
\]
When this quantity is available in closed form, exact normalization gives the cleanest procedure. In more complicated models, one may instead use a positive, predictable, certified upper bound
\begin{equation}
  \widetilde Z_i(X_i)\ge\ZhD{X_i}
  \label{eq:safe-normalizer-bound}
\end{equation}
computed after observing $X_i$ but before observing $Y_i$. Then
\[
  \widetilde e_i=\frac{h(X_i,Y_i)}{\widetilde Z_i(X_i)}
\]
satisfies
\[
  \E[\widetilde e_i\mid\cG_i]
  =\frac{\ZhD{X_i}}{\widetilde Z_i(X_i)}\le1
\]
under the source predictive null. Its running product is therefore a nonnegative supermartingale. Relative to exact normalization, the one-step log-evidence loss is
\begin{equation}
  \log e_i-\log\widetilde e_i
  =\log\frac{\widetilde Z_i(X_i)}{\ZhD{X_i}}\ge0.
  \label{eq:normalizer-log-penalty}
\end{equation}
A loose upper bound is safe but may substantially delay confirmation. Unless equality holds, the approximate factor is no longer the exact likelihood ratio $\phD{Y_i}{X_i}/\pzeroD{Y_i}{X_i}$; its log wealth is the exact cumulative log-score advantage minus the accumulated normalizer penalty.

The direction of approximation is essential. A denominator smaller than $\ZhD{X_i}$ makes the conditional mean exceed one. An ordinary unbiased Monte Carlo estimate is not generally safe either. If a positive estimate $\widehat Z_i$ is conditionally independent of $Y_i$ given $\cG_i$ and satisfies $\E[\widehat Z_i\mid\cG_i]=\ZhD{X_i}$, then Jensen's inequality gives
\[
  \E\!\left[\frac{h(X_i,Y_i)}{\widehat Z_i}\,\middle|\,\cG_i\right]
  =\ZhD{X_i}\E\!\left[\frac1{\widehat Z_i}\,\middle|\,\cG_i\right]\ge1,
\]
with strict inequality unless $\widehat Z_i=\ZhD{X_i}$ almost surely. Reciprocal refutation also requires exact normalization. Under the strict-positivity condition used for reciprocal refutation in \cref{sec:two-boundary}, the corrected predictive null gives
\[
  \E\!\left[\frac{\widetilde Z_i(X_i)}{h(X_i,Y_i)}\,\middle|\,\cG_i\right]
  =\frac{\widetilde Z_i(X_i)}{\ZhD{X_i}}\ge1.
\]
Further numerical details are given in \cref{app:normalizer-numerics}; all experiments use closed-form Gaussian normalizers.

\begin{algorithm}[t]
\caption{Sequential evidence for a prespecified predictive correction}
\label{alg:conditional}
\begin{algorithmic}[1]
\Require Source predictive $\pzeroD{y}{x}$, fixed correction $h(x,y)$, level $\alpha$, and either the exact normalizer or a predictable certified upper bound
\State $S_0\gets0$
\For{$i=1,2,\ldots$}
  \State Observe $X_i$ and compute $D_i^{\rm use}=\ZhD{X_i}$ in exact mode or $D_i^{\rm use}=\widetilde Z_i(X_i)$ in conservative mode
  \State Observe $Y_i$
  \State Compute $\ell_i\gets\log h(X_i,Y_i)-\log D_i^{\rm use}$
  \State Update $S_i\gets S_{i-1}+\ell_i$
  \If{$S_i>\log(1/\alpha)$}
    \State \textbf{stop:} reject the source predictive null and report relative evidence in the direction of the correction
  \ElsIf{exact two-boundary mode and $S_i<\log\alpha$}
    \State \textbf{stop:} reject the corrected predictive null and refute the correction relative to the source
  \EndIf
\EndFor
\end{algorithmic}
\end{algorithm}

In \cref{alg:conditional}, exact mode has $S_i=\log M_i$, the cumulative
log-score advantage of $\phDobj$ over $\pzeroDobj$. An upper crossing rejects $H_0^{\mathrm{pred}}$ and gives anytime-valid relative confirmation of the corrected predictive; a lower crossing rejects $H_h^{\mathrm{pred}}$ and refutes it relative to the source. With a certified upper-bound normalizer, an upper crossing still rejects the source predictive null, but the accumulated wealth is conservative directional evidence and no lower refutation boundary is available. The algorithm concerns only corrections to $Y\mid X$; pure covariate shift requires separate input-stream methods \citep{Choi2026arxiv_avccs}.

\subsection{Predictable Corrections}
\label{sec:predictable}

The fixed-correction setting is the cleanest for interpretation.  Validity also permits corrections chosen predictably.

\bigskip
\begin{proposition}[Predictable tilts]
\label{prop:predictable}
At time $i$, suppose that after observing $X_i$ but before observing $Y_i$, the practitioner chooses a nonnegative function \(h_i(X_i,\cdot)\) that is \(\cG_i\)-measurable.  Let
\[
  Z_i(X_i)=\int h_i(X_i,y)\pzeroD{y}{X_i}\,dy
\]
be finite and positive, and define
\[
  e_i=\frac{h_i(X_i,Y_i)}{Z_i(X_i)}.
\]
Then $e_i$ is a conditional e-value under $H_0^{\mathrm{pred}}$, and $\prod_{i=1}^t e_i$ is a nonnegative martingale under $H_0^{\mathrm{pred}}$, hence an e-process.
\end{proposition}

\begin{proof}[Proof sketch]
After conditioning on $\cG_i$, the predictable tilt is fixed as a function of the yet-unobserved outcome, so the same normalization argument applies.  See \cref{app:proof-predictable}.
\end{proof}

Predictable updating can be useful for adaptive betting or safe model
monitoring, but it changes the inferential object.  A fixed $h$ is designed to
confirm one prespecified correction.  A predictable sequence $h_i$ instead
shows that an adaptive betting strategy has accumulated evidence against the
source predictive null; without additional precommitment, it does not confirm
any single correction selected after observing the stream.

\section{Structured Conditional Predictive Corrections}
\label{sec:special-cases}

In this section, we instantiate the general construction of \cref{sec:eprocess-construction}
for several structured predictive corrections.  The common principle is that scientific or
operational knowledge specifies, before monitoring, how the source conditional predictive
$\pzeroD{y}{x}$ should be modified.  The resulting e-process then evaluates whether that
particular corrected predictive outpredicts its prespecified reference on the incoming target
stream.  The procedure does not estimate an unrestricted target distribution, and a crossing
does not by itself identify the physical mechanism responsible for the evidence.

The corrections considered here act on $Y\mid X$.  A label-shift assumption induces a
particular correction of the conditional label predictive, whereas concept drift motivates
direct corrections to the conditional response distribution.  Pure covariate shift changes
the marginal input distribution while leaving $Y\mid X$ unchanged and therefore requires a
separate input-stream construction, as developed in the covariate-balance paper
\citep{Choi2026arxiv_avccs}.  A full joint-shift analysis would combine input-distribution
evidence with the conditional predictive evidence studied here.

The subsections serve complementary purposes.  \Cref{sec:label-shift} derives the
label-shift-induced correction.  \Cref{sec:concept-drift-correction} develops Gaussian mean
and variance corrections and shows how evidence behaves when their magnitudes are
misspecified.  \Cref{sec:exp-tilt} gives a unifying exponential-tilt representation.
\Cref{sec:mixtures} handles prespecified uncertainty over the correction, while
\cref{sec:beyond-tolerance-correction} changes the decision problem by replacing the source
reference with an operational tolerance boundary.  Finally, \cref{sec:cross-family} studies
what the structured wealth processes do when the actual target change belongs to a different
mechanism family.  \Cref{tab:shift-correction-map} summarizes the inferential role of each
construction.

\begin{table}[ht!]
\centering
\small
\caption{Representative structured predictive corrections and their operational
interpretations.  The first two rows specify a corrected conditional predictive directly;
the mixture construction aggregates prespecified correction paths, and the tolerance
construction deliberately changes the reference and the null hypothesis.}
\label{tab:shift-correction-map}
\setlength{\tabcolsep}{4pt}
\begin{tabular}{p{0.16\linewidth}p{0.28\linewidth}p{0.26\linewidth}p{0.22\linewidth}}
\toprule
Setting or construction & Structural premise & Inferential object & Operational question \\
\midrule
Label shift &
$P_t^Y\ne P_s^Y$, with $P_t^{X\mid Y}=P_s^{X\mid Y}$ &
Label tilt $h(x,y)=w(y)$, inducing $\phD{y}{x}$ &
Deploy $\phDobj$ or retain $\pzeroDobj$? \\
\addlinespace
Concept drift &
$P_t^{Y\mid X}\ne P_s^{Y\mid X}$ &
Tilt encoding a mean, variance, subgroup, or other response correction &
Apply the structured correction or retain the source predictive? \\
\addlinespace
Mixture over corrections &
A correction family is prespecified, but its index is uncertain &
Weighted mixture $M_t^{\rm mix}=\int M_t(\theta)\,d\Pi(\theta)$ of full wealth paths &
Has the prespecified family accumulated global evidence against the source? \\
\addlinespace
Beyond-tolerance comparison &
A tolerated region and an actionable design point are prespecified &
Ratio $p_{\rm alarm}(y\mid x,\Dtr)/p_{\rm tol}(y\mid x,\Dtr)$ &
Is there sufficient evidence to act beyond the tolerated region? \\
\bottomrule
\end{tabular}
\end{table}

\subsection{Label-Shift Correction}
\label{sec:label-shift}

Under label shift, the conditional input distribution given the label is stable,
\[
  P_t^{X\mid Y}=P_s^{X\mid Y},
\]
while the label marginal changes, $P_t^Y\ne P_s^Y$.  Let $w(y)\ge0$ be a prespecified label weight with a finite, positive normalizer under
$\pzeroD{\cdot}{x}$ at every relevant input.  The corresponding tilt and corrected predictive are
\begin{equation}
  h(x,y)=w(y),
  \qquad
  \phD{y}{x}
  =
  \frac{w(y)\pzeroD{y}{x}}
       {\E_{Y\sim\pzeroD{\cdot}{x}}[w(Y)]}.
  \label{eq:label-shift-correction}
\end{equation}
Multiplying $w$ by a positive constant leaves $\phDobj$ unchanged because that constant
cancels in the normalizer.  The one-step e-value is
\begin{equation}
  e_i
  =
  \frac{w(Y_i)}
       {\E_{Y\sim\pzeroD{\cdot}{X_i}}[w(Y)]}.
  \label{eq:label-shift-evalue}
\end{equation}

If the source predictive equals the source conditional law $P_s^{Y\mid X}$, the target
satisfies exact label shift, and $P_t^Y\ll P_s^Y$, choosing
\[
  w(y)=\frac{dP_t^Y}{dP_s^Y}(y)
\]
recovers the target conditional law $P_t^{Y\mid X}$ through Bayes' rule.  If
$\pzeroDobj$ is instead a fitted or posterior predictive approximation to
$P_s^{Y\mid X}$, the same weighting still defines a valid prespecified predictive
correction, but it need not equal the exact target conditional distribution.  The e-process
assesses the induced corrected predictive itself; it does not require the label-shift model
to be exactly correct.  This specializes the general framework to anytime-valid confirmation
of a prespecified label-shift correction \citep{Choi2026testing}.

\paragraph{Deploy-or-retain decision.}
An external study, a known intervention, historical information, or a planned change in the
target population may suggest $w(y)$ before target outcomes are observed.  The operational
choice is whether to retain $\pzeroDobj$ or deploy the induced $\phDobj$.  This choice is
especially relevant when target labels are expensive, delayed, or revealed sequentially
\citep{LiptonZ2018icml,AlexandariAM2020icml,GargS2020neurips}.  An upper-boundary crossing
provides anytime-valid relative evidence for deploying $\phDobj$ over $\pzeroDobj$.  Failure
to cross is inconclusive: the correction may not be predictively preferable, or the observed
labels may simply be insufficiently informative.  By \cref{sec:growth}, positive expected
log-growth requires only that the induced corrected predictive be closer to the actual target
predictive than the source predictive is in conditional KL divergence.  The proposed label
weights therefore need not coincide with the exact target label ratio to accumulate positive
evidence.

\subsection{Concept-Drift Corrections}
\label{sec:concept-drift-correction}

Here concept drift refers to a change in the conditional response distribution,
\[
  P_t^{Y\mid X}\ne P_s^{Y\mid X}.
\]
A tilt $h(x,y)$ can encode a prespecified modification of this conditional distribution.
The setting is most useful when an intervention, protocol change, new deployment site, or
engineering analysis suggests a particular form of change before monitoring begins
\citep{QinSJ2012arc,KellyCJ2019bmcm,SubbaswamyA2020biostatistics}.  The e-process then asks
whether that proposed correction predicts the target outcomes better than retaining the
source predictive.  It is not a generic detector that searches the observed target stream
for an unknown form of concept drift.

The following Gaussian examples separate two common operational questions: whether to shift
the conditional center and whether to widen or narrow the conditional predictive
uncertainty.  They also make the relative nature of confirmation explicit: a correction may
accumulate positive evidence even when its magnitude is not exactly correct.

\subsubsection{Conditional Mean Correction}
\label{sec:mean-shift}

Suppose the source predictive is
\[
  \pzeroD{y}{x}
  =
  \mathcal N\!\left(
    y\mid \mu_0(x),\sigma_0^2(x)
  \right),
\]
and the proposed correction shifts the conditional mean along a known shape $g(x)$ by a
prespecified coefficient $\delta$:
\[
  \phD{y}{x}
  =
  \mathcal N\!\left(
    y\mid \mu_0(x)+\delta g(x),\sigma_0^2(x)
  \right).
\]
The one-step log e-value is
\begin{equation}
  \log e_i
  =
  \frac{\delta g(X_i)\{Y_i-\mu_0(X_i)\}}{\sigma_0^2(X_i)}
  -
  \frac{\delta^2g^2(X_i)}{2\sigma_0^2(X_i)}.
  \label{eq:mean-shift-e}
\end{equation}
Thus the process compares the proposed mean-corrected predictive with the source predictive
along the prespecified direction $g$.  A constant $g(x)=1$ gives a common additive offset,
whereas a nonconstant $g$ permits the correction to vary across subgroups, doses,
instruments, or other scientifically meaningful input characteristics.

\paragraph{Deciding whether to apply a directional offset.}
Instrument recalibration, a bridging experiment, or simulator-to-reality analysis may
suggest the offset $\delta g(x)$ before new outcomes arrive
\citep{WorkmanJJ2018as,KennedyMC2001jrsssb}.  The operational choice is whether to retain
$\mu_0(x)$ or deploy $\mu_0(x)+\delta g(x)$.  Sequential evidence is useful when calibration
outcomes arrive one at a time or when data collection may stop as soon as the proposed
adjustment is sufficiently supported.

\paragraph{Magnitude mismatch.}
Suppose the actual target predictive is Gaussian with conditional mean
$\mu_0(x)+\delta^*g(x)$ and variance $\sigma_0^2(x)$, whereas the proposed correction uses
$\delta$.  The conditional drift is
\begin{equation}
  \Gamma_\delta(x)
  =
  \frac{g^2(x)}{\sigma_0^2(x)}
  \left(
    \delta\delta^*-\frac{\delta^2}{2}
  \right)
  =
  \frac{g^2(x)}{2\sigma_0^2(x)}
  \delta(2\delta^*-\delta).
  \label{eq:mean-mismatch-drift}
\end{equation}
At an informative input, $g(x)\ne0$, positive drift is therefore equivalent to
\[
  \delta(2\delta^*-\delta)>0.
\]
For the common case $\delta^*>0$ with a proposed positive correction, this reduces to
\[
  0<\delta<2\delta^*.
\]
A correction with $0<\delta<\delta^*$ underestimates the true shift but still improves on the
source predictive.  A correction with $\delta^*<\delta<2\delta^*$ overestimates the shift but
remains closer to the target mean than the source mean does.  When $\delta>2\delta^*$, the
proposed offset overshoots so severely that it is worse in expected log score than applying
no correction.  The drift is maximized at $\delta=\delta^*$.

The factor
\[
  \frac{g^2(x)}{\sigma_0^2(x)}
\]
is the local information scale for this comparison.  Inputs at which the proposed mean
change is large relative to the predictive variance accumulate evidence more rapidly.  This
connects the structured correction directly to the adaptive-design result studied in
\cref{sec:exp-adaptive}.

\subsubsection{Conditional Variance Correction}
\label{sec:variance-shift}

Under the same Gaussian source predictive, suppose the conditional mean is retained while
the variance is multiplied by a prespecified factor $c>0$:
\[
  \phD{y}{x}
  =
  \mathcal N\!\left(
    y\mid \mu_0(x),c\sigma_0^2(x)
  \right).
\]
Then
\begin{equation}
  \log e_i
  =
  -\frac12\log c
  +
  \frac12\left(1-\frac1c\right)
  \frac{\{Y_i-\mu_0(X_i)\}^2}{\sigma_0^2(X_i)}.
  \label{eq:variance-shift-e}
\end{equation}
For $c>1$, large standardized residuals favor variance inflation.  For $0<c<1$, small
standardized residuals favor variance contraction.

\paragraph{Deciding whether predictive uncertainty should be widened or narrowed.}
A change in assay protocol, laboratory batch, sensor precision, or operating conditions may
leave the conditional mean approximately stable while changing response variability
\citep{JohnsonWE2007biostatistics,LeekJT2010nrg}.  A prespecified factor $c$ then represents
an operational proposal to widen or narrow the predictive distribution.  Relative
confirmation may support revised predictive intervals, quality-control limits, or downstream
risk thresholds while controlling false confirmation under the source predictive null.

\paragraph{Magnitude mismatch and mechanism ambiguity.}
Suppose the actual target predictive is Gaussian with the same conditional mean and variance
$c^*\sigma_0^2(x)$.  The conditional drift of a proposed factor $c$ is
\begin{equation}
  \Gamma_c(x)
  =
  -\frac12\log c
  +
  \frac12\left(1-\frac1c\right)c^*.
  \label{eq:variance-mismatch-drift}
\end{equation}
For $c\ne1$, define
\[
  \rho(c)
  \defeq
  \frac{c\log c}{c-1}.
\]
Then
\begin{equation}
  \Gamma_c(x)>0
  \quad\Longleftrightarrow\quad
  \begin{cases}
    c^*>\rho(c), & c>1,\\[2pt]
    c^*<\rho(c), & 0<c<1.
  \end{cases}
  \label{eq:variance-mismatch-condition}
\end{equation}
Moreover,
\[
  1<\rho(c)<c \quad\text{when }c>1,
  \qquad
  c<\rho(c)<1 \quad\text{when }0<c<1.
\]
Thus a variance-inflation proposal may overstate the actual inflation and still outpredict the
source, and a variance-contraction proposal may similarly overstate the contraction while
remaining predictively preferable.  Correct specification, $c=c^*$, maximizes expected
log-growth over $c>0$, but exact specification is not required for positive drift.

The interpretation is nevertheless predictive rather than mechanistic.  The e-value in
\eqref{eq:variance-shift-e} is driven by squared residuals, which can be enlarged by a mean
shift, heavy tails, outliers, or other misspecification as well as by a genuine variance
increase.  A crossing therefore favors the variance-corrected predictive over the source
predictive; it does not establish variance change as the unique cause.  The cross-family
calculations in \cref{sec:cross-family} quantify this limitation.

\subsection{General Exponential-Family Predictive Tilts}
\label{sec:exp-tilt}

The preceding examples are instances of a common exponential-tilt construction.  Let
$\phi(x,y)\in\R^d$ be a prespecified vector of interpretable features and let
$\eta\in\R^d$ be a prespecified correction coefficient.  Define
\begin{equation}
  h_\eta(x,y)
  =
  \exp\{\eta^\top\phi(x,y)\},
  \qquad
  Z_\eta(x)
  =
  \E_{Y\sim\pzeroD{\cdot}{x}}
  \left[
    \exp\{\eta^\top\phi(x,Y)\}
  \right],
  \qquad
  \psi_x(\eta)=\log Z_\eta(x).
  \label{eq:exp-family-tilt}
\end{equation}
Whenever $Z_\eta(x)$ is finite and positive, the corrected predictive is
\begin{equation}
  p_\eta(y\mid x,\Dtr)
  =
  \exp\{\eta^\top\phi(x,y)-\psi_x(\eta)\}
  \pzeroD{y}{x},
  \label{eq:exp-family-predictive}
\end{equation}
and the one-step log e-value is
\begin{equation}
  \log e_i
  =
  \eta^\top\phi(X_i,Y_i)-\psi_{X_i}(\eta).
  \label{eq:exp-family-log-e}
\end{equation}
The feature vector $\phi$ determines which aspects of the predictive distribution are
modified, whereas $\eta$ determines the proposed direction and magnitude in that feature
space.  In the primary confirmatory interpretation, both are fixed before target outcomes
are observed.  Predictable updates are valid under \cref{sec:predictable}, but then the
result concerns an adaptive betting strategy rather than one fixed correction.

\paragraph{Label-shift correction.}
For categorical $Y\in\{1,\ldots,K\}$, take
\[
  \phi_k(x,y)=\mathbf 1\{y=k\},
  \qquad k=1,\ldots,K,
\]
and set $\eta_k=\log w_k$ for positive class weights $w_k$.  Then
\[
  h_\eta(x,y)=w_y,
  \qquad
  p_\eta(y\mid x,\Dtr)
  =
  \frac{w_y\pzeroD{y}{x}}
       {\sum_{k=1}^K w_k\pzeroD{k}{x}}.
\]
This recovers \cref{eq:label-shift-correction}.  Adding the same constant to every
$\eta_k$, equivalently multiplying every $w_k$ by the same positive factor, leaves the
corrected predictive unchanged; only relative class weights are identifiable.

\paragraph{Conditional mean correction.}
For the Gaussian source predictive, take
\[
  \phi_{\rm mean}(x,y)
  =
  \frac{g(x)\{y-\mu_0(x)\}}{\sigma_0^2(x)}
\]
and $\eta=\delta$.  Then
\[
  Z_\delta(x)
  =
  \exp\!\left\{
    \frac{\delta^2g^2(x)}{2\sigma_0^2(x)}
  \right\},
\]
and normalization gives
\[
  p_\delta(y\mid x,\Dtr)
  =
  \mathcal N\!\left(
    y\mid \mu_0(x)+\delta g(x),\sigma_0^2(x)
  \right).
\]
Thus the Gaussian mean correction is an exponential tilt in a variance-scaled residual.

\paragraph{Conditional variance correction.}
Under the same source predictive, take
\[
  \phi_{\rm var}(x,y)
  =
  \frac{\{y-\mu_0(x)\}^2}{2\sigma_0^2(x)},
  \qquad
  \eta=1-\frac1c.
\]
Then
\[
  h_c(x,y)
  =
  \exp\!\left\{
    \left(1-\frac1c\right)
    \frac{\{y-\mu_0(x)\}^2}{2\sigma_0^2(x)}
  \right\},
  \qquad
  Z_c(x)=\sqrt c,
\]
and the normalized predictive is
\[
  p_c(y\mid x,\Dtr)
  =
  \mathcal N\!\left(
    y\mid \mu_0(x),c\sigma_0^2(x)
  \right).
\]
Thus the variance correction is an exponential tilt in the squared standardized residual.

\paragraph{Subgroup-specific and combined corrections.}
Interactions between response features and prespecified input indicators produce localized
corrections.  For example, if $\mathcal A_1,\ldots,\mathcal A_J$ are prespecified subgroups,
features
\[
  \phi_j(x,y)
  =
  \mathbf 1\{x\in\mathcal A_j\}
  \frac{y-\mu_0(x)}{\sigma_0^2(x)}
\]
with coefficients $\eta_j$ encode subgroup-specific mean offsets.  Interactions between
class and subgroup indicators similarly encode subgroup-specific label corrections.  A
feature vector containing both linear and quadratic residual terms can encode a joint
mean-and-variance correction, and dose, treatment, instrument, or batch variables can enter
through prespecified interactions.

\subsection{Prespecified Mixtures over Correction Uncertainty}
\label{sec:mixtures}

A practitioner may know the broad form of a correction while remaining uncertain about its
magnitude, direction, or mechanism index.  Let $\{h_\theta:\theta\in\Theta\}$ be a
prespecified family, let $\Pi$ be a probability measure on $\Theta$ fixed before monitoring,
and define
\[
  M_t(\theta)=\prod_{i=1}^t e_i(\theta).
\]

\begin{proposition}[Prespecified mixture and correction panel]
\label{prop:correction-panel}
Assume that every $h_\theta$ satisfies \eqref{eq:normalizer} and that
$(\theta,\omega)\mapsto M_t(\theta)(\omega)$ is jointly measurable with respect
to $\mathcal B(\Theta)\otimes\cF_t$. By \cref{thm:eprocess}, for every
$P\in\mathcal P_0^{\mathrm{pred}}$ and every $\theta$, the process
$(M_t(\theta))_{t\ge0}$ is a nonnegative $P$-martingale with
$M_0(\theta)=1$. Then
\begin{equation}
  M_t^{\mathrm{mix}}
  =
  \int M_t(\theta)\,d\Pi(\theta)
  \label{eq:mixture-eprocess}
\end{equation}
is a nonnegative martingale under every
$P\in\mathcal P_0^{\mathrm{pred}}$, with $M_0^{\mathrm{mix}}=1$, and hence is an
e-process for the source predictive null class.  For a finite panel
$\{h^{(1)},\ldots,h^{(K)}\}$ with prespecified weights $\pi_k\ge0$ satisfying
$\sum_k\pi_k=1$,
\[
  M_t^{\mathrm{panel}}
  =
  \sum_{k=1}^K\pi_k M_t^{(k)}
\]
therefore provides an anytime-valid family-level test at level $\alpha$.
\end{proposition}

\begin{proof}[Proof sketch]
For any $P\in\mathcal P_0^{\mathrm{pred}}$, conditional Tonelli's theorem and the component
martingale property give
\[
  \E_P[M_t^{\mathrm{mix}}\mid\cF_{t-1}]
  =
  \int \E_P[M_t(\theta)\mid\cF_{t-1}]\,d\Pi(\theta)
  =
  \int M_{t-1}(\theta)\,d\Pi(\theta)
  =
  M_{t-1}^{\mathrm{mix}}.
\]
See \cref{app:proof-mixture} for the full argument.
\end{proof}

The mixture in \eqref{eq:mixture-eprocess} averages complete wealth paths:
\[
  \int \prod_{i=1}^t e_i(\theta)\,d\Pi(\theta).
\]
It is generally different from the product of pointwise mixtures
\[
  \prod_{i=1}^t\int e_i(\theta)\,d\Pi(\theta).
\]
The first construction corresponds to assigning initial wealth across persistent correction
indices and retaining those indices through time.  It is the relevant object when the
uncertainty concerns which one of a prespecified set of corrections may be useful.

Each component is individually an anytime-valid test, but inspecting many components and
reporting whichever one crosses or attains the largest wealth does not control the resulting
familywise or post-selection claim at level $\alpha$.  A crossing of $M_t^{\mathrm{mix}}$
supports one global statement: the prespecified weighted correction family has accumulated
evidence against the source predictive null.  It does not identify a unique $\theta$, confirm
every component, or license an unadjusted claim about the data-selected best component.
Simultaneous or selected componentwise claims require an explicit error allocation or another
prespecified rule.

\paragraph{Correction uncertainty.}
When prior knowledge identifies the direction of a correction but leaves its
magnitude uncertain, a prespecified mixture can distribute evidence across a
set of plausible corrections without committing to a single one before
monitoring.  Choosing a particular correction for subsequent deployment,
however, is a separate post-confirmation selection or decision problem unless
the selection rule is itself prespecified.

\subsection{Beyond-Tolerance Confirmation}
\label{sec:beyond-tolerance-correction}

In many applications, the relevant question is not whether the target predictive differs at
all from the source, but whether the departure is large enough to justify action.  This is a
different inferential problem from the target-misspecification robustness analysis in
\cref{sec:misspecification-control}.  There the original null remains the source predictive and one asks
where its crossing guarantee happens to persist.  Here the practitioner deliberately defines
a new null representing an acceptable region of change.

Let $p_{\rm tol}(y\mid x,\Dtr)$ denote the predictive distribution at the largest acceptable
shift, and let $p_{\rm alarm}(y\mid x,\Dtr)$ denote a prespecified actionable design point
beyond that boundary.  Define
\begin{equation}
  e_i^{\rm tol}
  =
  \frac{p_{\rm alarm}(Y_i\mid X_i,\Dtr)}
       {p_{\rm tol}(Y_i\mid X_i,\Dtr)},
  \qquad
  M_t^{\rm tol}
  =
  \prod_{i=1}^t e_i^{\rm tol}.
  \label{eq:tolerance-e}
\end{equation}
The evidence is now anchored at the tolerated boundary rather than at the uncorrected source
predictive.  An upper crossing favors the actionable predictive over the tolerated-boundary
predictive on the observed target stream.

\paragraph{Acting only on practically meaningful change.}
Small deviations may be scientifically real but too small to justify recalibration, process
interruption, clinical review, or another costly intervention.  A tolerance policy therefore
specifies, before monitoring, both an acceptable region and an actionable design point.  The
question becomes ``Is there sufficient evidence to act beyond tolerance?'' rather than
``Has any change occurred?'' \citep{PodkopaevA2022iclr}.  A crossing remains a relative
predictive statement: it does not estimate the exact target parameter, prove that the target
has reached the nominal alarm design point, or identify a unique mechanism.

Treating \eqref{eq:tolerance-e} merely as a likelihood ratio with $p_{\rm tol}$ as reference
would control false alarms only at that single boundary distribution.  A genuine tolerance
policy should control false alarms throughout the entire acceptable region.  A regular
one-parameter exponential-tilt family provides this stronger composite-null guarantee.

\bigskip
\begin{proposition}[False-alarm control over a composite tolerated region]
\label{prop:tolerance}
Let $\phi:\cX\times\cY\to\R$ and
\[
  p_\eta(y\mid x,\Dtr)
  =
  \frac{\exp\{\eta\phi(x,y)\}\pzeroD{y}{x}}{Z_\eta(x)},
  \qquad
  \psi_x(\eta)=\log Z_\eta(x),
\]
define a one-parameter tilted family.  Assume that there is an open interval $\mathcal I$,
common to all relevant inputs, on which every $\psi_x$ is finite.  Fix
$\eta_{\rm tol}\in\mathcal I$ and
$\eta_{\rm alarm}=\eta_{\rm tol}+\Delta\in\mathcal I$ with $\Delta>0$, and set
\[
  e_i^{\rm tol}
  =
  \frac{p_{\eta_{\rm alarm}}(Y_i\mid X_i,\Dtr)}
       {p_{\eta_{\rm tol}}(Y_i\mid X_i,\Dtr)}
  =
  \exp\{\Delta\phi(X_i,Y_i)\}
  \frac{Z_{\eta_{\rm tol}}(X_i)}{Z_{\eta_{\rm alarm}}(X_i)}.
\]
Suppose that
\[
  Y_i\mid\cG_i
  \sim
  p_{\eta_i}(\cdot\mid X_i,\Dtr),
\]
where $\eta_i$ is $\cG_i$-measurable, takes values in $\mathcal I$, and satisfies
$\eta_i\le\eta_{\rm tol}$ almost surely for every $i$.  Then
\[
  \E[e_i^{\rm tol}\mid\cG_i]\le1.
\]
Let $\mathcal P_0^{\rm tol}$ denote the class of all data-stream distributions induced by
predictable sequences $(\eta_i)$ satisfying $\eta_i\le\eta_{\rm tol}$ almost surely for
every $i$, together with any admissible input process.  Consequently,
$(M_t^{\rm tol})_{t\ge0}$ is a nonnegative supermartingale under every
$P\in\mathcal P_0^{\rm tol}$ and
\[
  \sup_{P\in\mathcal P_0^{\rm tol}}
  \Prob_P\!\left\{
    \sup_{t\ge0}M_t^{\rm tol}>\frac1\alpha
  \right\}
  \le\alpha.
\]
Thus the deliberately specified composite null is the collection of target streams whose
conditional natural parameter never exceeds the tolerated boundary.
\end{proposition}

\begin{proof}[Proof sketch]
Conditional on $X_i=x$ and under $p_{\eta_i}$,
\[
  \E[e_i^{\rm tol}\mid\cG_i]
  =
  \exp\!\left\{
    \psi_x(\eta_i+\Delta)-\psi_x(\eta_i)
    -\psi_x(\eta_{\rm tol}+\Delta)+\psi_x(\eta_{\rm tol})
  \right\}.
\]
For a convex function, an increment of fixed length $\Delta$ is nondecreasing in its starting
point.  Since $\eta_i\le\eta_{\rm tol}$, the exponent is nonpositive.  See
\cref{app:proof-tolerance} for the full argument.
\end{proof}

\begin{remark}[Where evidence begins to favor action]
\label{rem:tolerance-midpoint}
Assume in addition that $\psi_x$ is differentiable.  Under $p_\eta$, the conditional drift at
input $x$ is
\begin{equation}
  \E[\log e_i^{\rm tol}\mid X_i=x]
  =
  \Delta\psi_x'(\eta)
  -
  \{\psi_x(\eta_{\rm tol}+\Delta)-\psi_x(\eta_{\rm tol})\}.
  \label{eq:tolerance-drift}
\end{equation}
The mean value theorem gives at least one
$\eta_{\rm mid}(x)\in(\eta_{\rm tol},\eta_{\rm alarm})$ satisfying
\[
  \psi_x'\{\eta_{\rm mid}(x)\}
  =
  \frac{\psi_x(\eta_{\rm tol}+\Delta)-\psi_x(\eta_{\rm tol})}{\Delta}.
\]
If $\psi_x$ is strictly convex, this point is unique and the drift is positive exactly when
$\eta>\eta_{\rm mid}(x)$.  Hence the parameter line has three operational regions:
$\eta\le\eta_{\rm tol}$ is the tolerated region with false-alarm control;
$\eta_{\rm tol}<\eta\le\eta_{\rm mid}(x)$ lies outside tolerance but still favors the
tolerated-boundary predictive in expected log score; and $\eta>\eta_{\rm mid}(x)$ gives
positive local evidence growth toward an alarm.  Thus the procedure need not wait until the
true parameter reaches $\eta_{\rm alarm}$, but not every departure just beyond tolerance has
positive drift.

When $\psi_x$ is a nonconstant quadratic function, as at informative inputs in the Gaussian
mean-shift family of \cref{sec:mean-shift},
\[
  \eta_{\rm mid}(x)
  =
  \frac{\eta_{\rm tol}+\eta_{\rm alarm}}{2}.
\]
The indifference region is therefore a structural consequence of comparing two separated
predictive design points, not a peculiarity of one numerical example.
\end{remark}

The Gaussian mean-shift case is studied in \cref{sec:exp-tolerance}.
\Cref{prop:tolerance} is deliberately one-dimensional.  For a vector parameter
$\eta\in\R^d$ and a fixed alarm direction $\Delta$, convexity gives an analogous ordering
along the ray $\{\eta_{\rm tol}-s\Delta:s\ge0\}$, but
$\eta\mapsto\psi_x(\eta+\Delta)-\psi_x(\eta)$ need not define a monotone half-space over all
of $\R^d$.

\subsection{Cross-Family Drift Calculus}
\label{sec:cross-family}

The drift formula \eqref{eq:gamma-x} can be evaluated under target distributions
that do not belong to the structural family used to construct the e-process.
Together with the protected-half-space analysis in
\cref{sec:protected-halfspace}, these calculations clarify
a central interpretive limitation: a crossing supports the chosen corrected predictive
relative to its reference, but the statistic used by that correction may also respond to a
different physical mechanism.  The following results organize the stress tests in
\cref{sec:exp-cross-family}.

\begin{remark}[Convex tilts and mean-preserving spreads]
\label{rem:convex}
Several corrections above use tilts that are convex in $y$.  The mean correction uses the
exponential of a linear function, and a variance-inflation correction uses the exponential of
a positive quadratic.  If $q(\cdot\mid x)$ is a mean-preserving spread of
$\pzeroD{\cdot}{x}$, then convex ordering gives
\[
  \E_{Y\sim q(\cdot\mid x)}[h(x,Y)]
  \ge
  \E_{Y\sim\pzeroD{\cdot}{x}}[h(x,Y)],
\]
with strict inequality for a nontrivial spread when the tilt is strictly convex and the
expectations are finite; the left-hand side may also be infinite.  The moment condition
\eqref{eq:moment-condition} can therefore fail under dispersion or tail inflation, so the
source-null e-process need not retain false-confirmation control under that misspecified
target, regardless of whether the misspecification belongs to the same
structural family as the proposed correction.
\end{remark}

\subsubsection{Variance Correction under an Arbitrary Target}
\label{sec:cross-family-variance}

Let
\[
  m_2(x)
  =
  \E_q[\{Y-\mu_0(x)\}^2\mid X=x]
\]
be the target conditional second moment about the source mean.  For the variance correction
in \eqref{eq:variance-shift-e},
\begin{equation}
  \Gamma_c(x)
  =
  -\frac12\log c
  +
  \frac12\left(1-\frac1c\right)
  \frac{m_2(x)}{\sigma_0^2(x)},
  \qquad
  \Gamma_c(x)>0
  \iff
  \begin{cases}
    \displaystyle
    \frac{m_2(x)}{\sigma_0^2(x)}
    >
    \frac{c\log c}{c-1}, & c>1,\\[0.9em]
    \displaystyle
    \frac{m_2(x)}{\sigma_0^2(x)}
    <
    \frac{c\log c}{c-1}, & 0<c<1.
  \end{cases}
  \label{eq:var-drift-general}
\end{equation}
For $c>1$, any mechanism that increases the second moment about the source mean beyond the
threshold produces positive drift.  In particular, a pure conditional mean shift
$\mu_q(x)=\mu_0(x)+\delta^*g(x)$ with unchanged target conditional variance gives
\[
  m_2(x)
  =
  \sigma_0^2(x)+\delta^{*2}g^2(x),
\]
so the variance-inflation process has positive drift whenever
\begin{equation}
  \frac{\delta^{*2}g^2(x)}{\sigma_0^2(x)}
  >
  \frac{c\log c}{c-1}-1.
  \label{eq:cross-family-threshold}
\end{equation}
For $c=1.8$, the right-hand side is approximately $0.3225$.  Thus a variance-process crossing
can be driven by a mean shift even when the conditional variance has not changed.

Conversely, a heavier-tailed target with the same second moment as the source leaves the
long-run drift negative for $c>1$, but this does not restore an anytime-valid crossing bound.
Rare large residuals may still produce an early boundary crossing when the moment condition
of \cref{prop:misspecification-control} fails.  A variance-process crossing should therefore
be interpreted as evidence for the variance-corrected predictive, not as identification of
variance inflation; whether deployment is scientifically appropriate may require checking
plausible mean-shift and tail-change explanations.

\subsubsection{Mean Correction under an Arbitrary Target}
\label{sec:cross-family-mean}

For the mean correction in \eqref{eq:mean-shift-e}, let $\mu_q(x)=\E_q[Y\mid X=x]$.  Then
\begin{equation}
  \Gamma_\delta(x)
  =
  \frac{\delta g(x)\{\mu_q(x)-\mu_0(x)\}}{\sigma_0^2(x)}
  -
  \frac{\delta^2g^2(x)}{2\sigma_0^2(x)}.
  \label{eq:mean-drift-general}
\end{equation}
The drift depends on $q$ only through its conditional mean.  Hence a mean-preserving target
change gives nonpositive drift and gives strictly negative drift whenever
$\delta g(x)\ne0$.  This drift calculation does not imply anytime-valid protection against
all mean-preserving changes.  As formalized in \cref{rem:convex}, the one-step
factor is an exponential of a linear residual and is therefore convex; a mean-preserving spread can make its conditional mean exceed one even
though its expected log is negative.  Dispersion or tail changes can consequently inflate
the maximal crossing probability without improving the long-run log-growth rate.

\paragraph{Practical diagnostic implication.}
Cross-family calculations are stress tests for interpretation, not alternative confirmatory
guarantees.  A variance-process crossing may be generated by a mean shift, and a mean-process
crossing may be made more frequent by dispersion or tail changes despite negative long-run
drift.  When several mechanisms are scientifically plausible, their prespecified wealth
paths may be inspected diagnostically, but a formal family-level claim should use the
mixture construction of \cref{prop:correction-panel} or an explicit error allocation.  When
exact normalization and mutual absolute continuity hold, the lower boundary of
\cref{sec:refutation} can also refute a proposed correction relative to the source predictive.

\section{Synthetic Experiments}
\label{sec:experiments}

The experiments use an oracle Gaussian conditional so that the e-process is
isolated from estimation error and can be checked against exact analytic
predictions.  They verify the main interpretations of the construction: the
process is anytime-valid under the conditional predictive null, its growth
matches the drift calculus of \cref{sec:growth,sec:cross-family}, it can be
anchored at an operational tolerance boundary, adaptive input selection can
accelerate evidence accumulation, and the observed failure modes under target
misspecification occur only where the moment condition of
\cref{prop:misspecification-control} fails.  The goal is not to benchmark
distribution-shift estimation.

\subsection{Common setup and reproducibility protocol}
\label{sec:exp-setup}

The source predictive is the oracle Gaussian conditional
\[
  \pzeroD{y}{x}=\mathcal N\!\left( y \mid \mu_0(x),\sigma^2 \right),
  \qquad
  \mu_0(x)=\sin(1.5x)+0.3x,
  \qquad
  \sigma=0.8.
\]
Unless otherwise stated, source inputs satisfy $X\sim\mathcal N(0,1)$.  We use an oracle conditional distribution to isolate the e-process behavior from estimation error.  Unless otherwise stated, all tests use $\alpha=0.05$, horizon $T=200$, and $5000$ Monte Carlo replications.  Confirmation occurs when $\log M_t>\log(1/\alpha)$.

For a proposed conditional mean correction
$p_\delta(y\mid x,\Dtr)=\mathcal N\!\left( y \mid \mu_0(x)+\delta g(x),\sigma^2 \right)$, we use
\[
  g(x)=1+0.5\tanh(x),
  \qquad
  \log e_i
  =
  \frac{\delta g(X_i)\{Y_i-\mu_0(X_i)\}}{\sigma^2}
  -\frac{\delta^2 g^2(X_i)}{2\sigma^2}.
\]
For a proposed variance correction
$p_c(y\mid x,\Dtr)=\mathcal N\!\left( y \mid \mu_0(x),c\sigma^2 \right)$, we use
\[
  \log e_i
  =
  -\frac12\log c
  +\frac12\left(1-\frac1c\right)
  \frac{\{Y_i-\mu_0(X_i)\}^2}{\sigma^2}.
\]

\paragraph{Seeding protocol.}
A root \texttt{SeedSequence(20260707)} is spawned into one child per experiment family; within each family, a single input matrix and a single standardized-residual matrix are drawn once and shared across all conditions of that family (residuals are rescaled per condition).  Two conditions that are mathematically identical---for instance, the variance-correction null and the $c_{\rm mis}=1$ row of the miscalibration sweep---therefore produce \emph{identical} numbers by construction, rather than approximately equal numbers from independent streams.

\paragraph{Analytic cross-checks.}
Under this setup $\E[g^2(X)]\approx1.0986$ for $X\sim\mathcal N(0,1)$.  Every mean final log-wealth in \cref{tab:baseline,tab:miscal,tab:crossfam} agrees with the corresponding fixed-correction analytic drift prediction $\overline\Gamma\cdot T$ from \cref{sec:growth,sec:cross-family} to within Monte Carlo error; for example, the correctly specified mean correction has $\overline\Gamma=\delta^2\E[g^2]/(2\sigma^2)=0.1738$ and observed mean $\log M_T=34.65\approx0.1738\times200$, and the correctly specified variance correction has $\overline\Gamma=0.1061$ and observed $21.22=0.1061\times200$.  The rows of \cref{tab:robust} require separate checks because the mixture, predictable plug-in, and adaptive-design strategies do not share one fixed drift.  For the five-component uniform mixture under the power condition, the leading finite-mixture approximation $\log M_T^{\rm mix}\approx\log(1/5)+\max_{\theta}\overline\Gamma_\theta T$ gives $33.15$ nats, within $0.01$ nats of the observed $33.14$.

\begin{table}[t]
\centering
\caption{Baseline synthetic experiments ($T=200$, $\alpha=0.05$, $5000$ replications).  Confirmation rate is the fraction of replications in which $\log M_t$ crosses $\log(1/\alpha)$ by $T$.  With $5000$ replications, the binomial Monte Carlo standard error of any reported rate is at most $0.0071$.  Median stopping time is reported among confirmed replications; for null conditions this conditions on the rare false confirmations.}
\label{tab:baseline}
\medskip
\begin{tabular}{llrrr}
\toprule
Experiment & Method or condition & Confirm. & Median $\tau$ & Mean $\log M_T$\\
\midrule
Mean shift & correct $\delta=0.45$ & 1.000 & 16 & 34.65\\
Mean shift & underspecified $\delta=0.225$ & 1.000 & 23 & 26.01\\
Mean shift & overspecified $\delta=0.9$ & 0.818 & 16 & $-0.15$\\
Mean-shift null & test $\delta=0.45$ & 0.034 & 15 & $-34.81$\\
Variance shift & correct $c=1.8$ & 0.997 & 26 & 21.22\\
Variance shift & underspecified $c=1.4$ & 0.999 & 33 & 17.78\\
Variance shift & overspecified $c=2.4$ & 0.972 & 25 & 17.46\\
Variance-shift null & test $c=1.8$ & 0.027 & 20 & $-14.33$\\
\bottomrule
\end{tabular}
\end{table}

\begin{figure}[t]
\centering
\includegraphics[width=\linewidth]{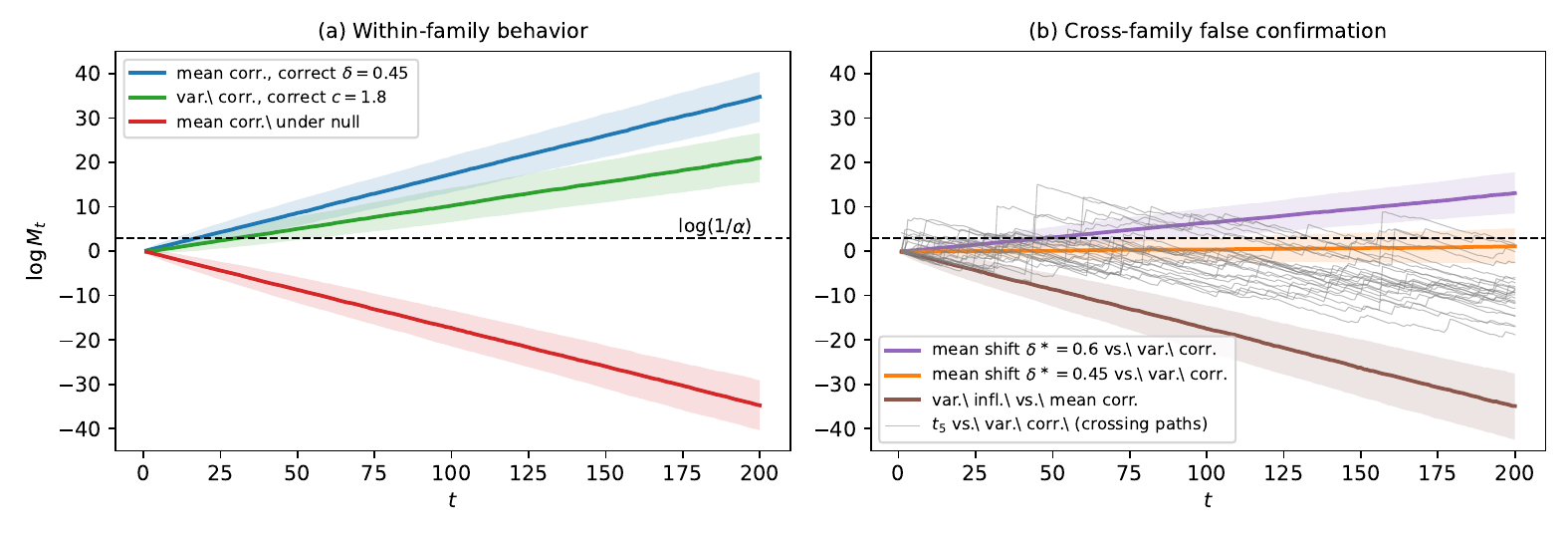}
\caption{Log-wealth trajectories (median and interquartile band over $5000$ replications; dashed line at $\log(1/\alpha)$).  (a) Within-family behavior: correctly specified mean and variance corrections grow linearly at the predicted drifts, while the null process drifts down.  (b) Cross-family false confirmation: a pure mean shift drives the variance-correction e-process upward at the drift predicted by \eqref{eq:var-drift-general}; a mean-preserving variance inflation drives the mean-correction e-process \emph{down} in median, yet individual paths can cross because the supermartingale condition fails; and under a heavy-tailed target (gray sample paths, $t_5$ with matched variance) crossings occur by single-observation jumps despite negative drift.}
\label{fig:paths}
\end{figure}

\subsection{Label-shift sanity check}

For Gaussian $\pzeroDobj$, the label tilt $h(y)=\exp(\beta y)$ is exactly equivalent to a global mean correction with $g\equiv1$ and $\delta=\beta\sigma^2$ (\cref{app:gaussian-tilt}).  With $\beta=0.55$ (so $\delta=0.352$), the maximum absolute difference between the cumulative label-tilt log-wealth paths and the corresponding mean-correction log-wealth paths, over $1000$ paths of length $200$, is below $2\times10^{-14}$.  This check uses $1000$ rather than the $5000$ paths of the other experiments because the comparison is a deterministic algebraic identity rather than a Monte Carlo estimate: the two implementations agree pathwise, so the number of paths affects only the number of opportunities to detect a coding discrepancy, not the precision of an estimate.  Thus the general predictive-correction implementation recovers the label-shift construction in this special case, up to floating-point accumulation.

\subsection{Conditional mean shift}
\label{sec:exp-mean-shift}

We generate
\[
  Y\mid X=x\sim
  \mathcal N\!\left( \mu_0(x)+0.45 g(x),\sigma^2 \right).
\]
The correctly specified mean correction confirms in all replications with median stopping time $16$ (\cref{fig:paths}a).  An underspecified correction $\delta=0.225$ remains powerful but slower, confirming in all replications with median stopping time $23$.  An overspecified correction $\delta=0.9$ has exactly zero log-drift under this target---by \eqref{eq:mean-drift-general}, $\Gamma_\delta\propto\delta^*\delta-\delta^2/2$ vanishes at $\delta=2\delta^*$---and correspondingly confirms in probability $0.818$ with mean final log wealth $-0.15$.  This illustrates the sensitivity of direct confirmation to correction magnitude.  Under the predictive null $(\delta^*=0)$, the same $\delta=0.45$ e-process confirms in probability $0.034$.

\paragraph{Checking the finite-horizon bound.}
This condition also calibrates \cref{cor:finite-horizon}.  For the correctly
specified correction the increments $\log e_i$ are i.i.d.\ with
$\overline\Gamma=\delta^2\E[g^2]/(2\sigma^2)=0.1738$ and, since
$\log e_i=\delta^2g^2(X_i)/(2\sigma^2)+\delta g(X_i)\varepsilon_i/\sigma$ with
$\varepsilon_i\sim\mathcal N(0,1)$ independent of $X_i$,
\[
  \Var(\log e_i)
  =
  \left\{\frac{\delta^2}{2\sigma^2}\right\}^2\Var\{g^2(X)\}
  +
  \frac{\delta^2}{\sigma^2}\E[g^2(X)]
  =
  0.3576 ,
\]
using $\E[g^2]=1.0986$ and $\E[g^4]=1.6073$.  With $b=\log 20=2.996$, the
variance-only bound \eqref{eq:finite-horizon-power} gives
$\Prob(\tau^*>t)\le0.686$, $0.356$, and $0.066$ at $t=30$, $50$, and $200$.
These are valid but very loose against an observed median stopping time of $16$
and a confirmation rate of $1.000$ by $t=200$, which is the expected behavior of
a variance-only bound whose decay is only of order
$\Var(\log e_i)/\{t\overline\Gamma^2\}$.  A certified sub-Gaussian proxy is also
available.  Since $g(x)\in(0.5,1.5)\subseteq[0.5,1.5]$, Hoeffding's lemma controls
the bounded $g^2(X_i)$ term. The Gaussian term is conditionally sub-Gaussian
with a variance proxy bounded uniformly in $X_i$. Consequently, if $A_i$
denotes the centered bounded term and $B_i$ the centered Gaussian term, then
the tower property gives
\[
  \E\exp\{\lambda(A_i+B_i)\}
  =
  \E\!\left[
    e^{\lambda A_i}
    \E\!\left(e^{\lambda B_i}\mid X_i\right)
  \right]
  \le
  \exp\!\left(\frac{\lambda^2s_B^2}{2}\right)
  \E e^{\lambda A_i}
  \le
  \exp\!\left(\frac{\lambda^2(s_A^2+s_B^2)}{2}\right).
\]
Thus the two proxies add in this particular dependent decomposition, yielding
\[
  s^2
  =
  \frac14\left\{
    \frac{\delta^2}{2\sigma^2}
    (1.5^2-0.5^2)
  \right\}^2
  +
  \frac{\delta^2}{\sigma^2}(1.5)^2
  =0.7369.
\]
Substitution into \eqref{eq:finite-horizon-subgaussian} gives the bounds
$0.895$, $0.644$, and $0.033$ at $t=30$, $50$, and $200$.  The certified proxy
is deliberately conservative and is therefore looser than Cantelli's bound at
the two shorter horizons, but its exponential decay becomes sharper by
$t=200$.

\subsection{Beyond-tolerance confirmation}
\label{sec:exp-tolerance}

In practice the actionable question may not be whether \emph{any} shift is present.  A deployment policy may tolerate shifts up to a boundary $\delta_{\rm tol}$ and ask for anytime-valid evidence only when the correction appears to exceed that tolerance.  In the one-sided Gaussian mean-shift family below, the correct likelihood ratio is not the actionable correction against the unshifted source, since that would also react to acceptable changes.  Instead, for an actionable level $\delta_{\rm alarm}>\delta_{\rm tol}$, we compare the actionable predictive to the tolerated-boundary predictive:
\[
  e_i^{\rm tol}
  =
  \frac{p_{\delta_{\rm alarm}}(Y_i\mid X_i,\Dtr)}
       {p_{\delta_{\rm tol}}(Y_i\mid X_i,\Dtr)}.
\]
For the Gaussian mean-correction family this gives
\[
  \log e_i^{\rm tol}
  =
  \frac{(\delta_{\rm alarm}-\delta_{\rm tol})g(X_i)\{Y_i-\mu_0(X_i)\}}{\sigma^2}
  -
  \frac{(\delta_{\rm alarm}^2-\delta_{\rm tol}^2)g^2(X_i)}{2\sigma^2}.
\]
This family is the exponential-family tilt of \cref{prop:tolerance} with
the standardized-residual feature
$\phi(x,y)=g(x)\{y-\mu_0(x)\}/\sigma^2$ of \cref{sec:exp-tilt} and
$\eta=\delta$, so that $\eta_{\rm tol}=\delta_{\rm tol}$ and
$\eta_{\rm alarm}=\delta_{\rm alarm}$ with no rescaling. Within the Gaussian
mean-shift family with unchanged conditional variance $\sigma^2$, the level is
controlled over the whole tolerated parameter set and not merely at its
boundary. Here the moment can be written in closed form: under a Gaussian
target with true mean-shift parameter $\delta^*$ and conditional variance
$\sigma^2$,
\begin{equation}
  \E_{\delta^*}\!\left[e_i^{\rm tol}\mid X_i=x,\cF_{i-1}\right]
  =\exp\!\left\{\frac{g^2(x)\,(\delta_{\rm alarm}-\delta_{\rm tol})(\delta^*-\delta_{\rm tol})}{\sigma^2}\right\},
  \label{eq:tolerance-moment}
\end{equation}
which is at most one for every $x$ exactly when
$\delta^*\le\delta_{\rm tol}$ (here $g(x)>0$ for every $x$). Therefore,
within this Gaussian family, $\prod_i e_i^{\rm tol}$ is a valid e-process
uniformly over the composite tolerated parameter regime
$\{\delta^*\le\delta_{\rm tol}\}$. The composite null is composite in the
mean-shift parameter only; the guarantee does not automatically extend to
target distributions outside this family, such as targets with an additional
variance or tail change. The test is intentionally not a sharp detector of
every $\delta^*>\delta_{\rm tol}$: because $\psi_x$ is quadratic here, \cref{rem:tolerance-midpoint} places the drift sign change exactly at the midpoint $(\delta_{\rm tol}+\delta_{\rm alarm})/2$, giving a practical indifference region between acceptable and clearly actionable shifts.

We set $\delta_{\rm tol}=0.25$ and $\delta_{\rm alarm}=0.55$, with the same source predictive, $T=200$, $\alpha=0.05$, and $5000$ replications as above.  \Cref{tab:tolerance-direct} shows that confirmation remains below $\alpha$ at and below the tolerance boundary, while becoming frequent once the true correction is clearly beyond tolerance.

\begin{table}[t]
\centering
\caption{Beyond-tolerance confirmation for the direct model-based e-process.  The tested likelihood ratio is $p_{\delta_{\rm alarm}}(\cdot\mid\cdot,\Dtr)/p_{\delta_{\rm tol}}(\cdot\mid\cdot,\Dtr)$ with $\delta_{\rm tol}=0.25$ and $\delta_{\rm alarm}=0.55$.  Median stopping time is reported among confirmed paths and is unstable when confirmation is extremely rare.}
\label{tab:tolerance-direct}
\medskip
\begin{tabular}{rrrr}
\toprule
True $\delta^*$ & Confirmation & Median $\tau$ & Mean $\log M_T$\\
\midrule
0.00 & $<0.001$ & 12 & $-41.02$\\
0.15 & 0.004 & 17.5 & $-25.87$\\
0.25 & 0.039 & 32 & $-15.42$\\
0.35 & 0.271 & 52 & $-5.25$\\
0.45 & 0.830 & 54 & 5.23\\
0.55 & 0.996 & 33 & 15.41\\
0.70 & 1.000 & 19 & 30.96\\
\bottomrule
\end{tabular}
\end{table}

This experiment gives the direct predictive-correction framework a policy interpretation.  The practitioner can prespecify a tolerance boundary, choose an actionable alternative beyond it, and obtain an anytime-valid alarm for evidence favoring the actionable correction over the tolerated one.  The price is the usual likelihood-ratio geometry: there is an indifference region between the two design points, and larger separation between $\delta_{\rm tol}$ and $\delta_{\rm alarm}$ gives a more conservative alarm near the boundary.

\subsection{Conditional variance shift and source miscalibration}
\label{sec:exp-variance}

We generate
\[
  Y\mid X=x\sim
  \mathcal N\!\left( \mu_0(x),c_{\rm mis}\sigma^2 \right),
\]
and run the variance-correction e-process with model variance $\sigma^2$.  The correctly specified correction ($c=c_{\rm mis}=1.8$) confirms in probability $0.997$ with median stopping time $26$.  Magnitude-misspecified corrections $c=1.4$ and $c=2.4$ also confirm frequently ($0.999$ and $0.972$) but with smaller average final log wealth ($17.78$ and $17.46$ versus $21.22$), matching the drifts $0.0889$ and $0.0873$ from \eqref{eq:var-drift-general}.  Under the predictive null ($c_{\rm mis}=1$), the $c=1.8$ correction confirms in probability $0.027$.

The same construction doubles as a stress test for source predictive miscalibration: if the deployment distribution equals the data-level source but the fitted predictive variance is too small by the factor $c_{\rm mis}$, the model-based predictive null is false, and the $c=1.8$ e-process confirms at the rates in \cref{tab:miscal}.  Because all rows share the same standardized residuals, the $c_{\rm mis}=1$ row \emph{is} the variance-null row and the $c_{\rm mis}=1.8$ row \emph{is} the correctly specified row of \cref{tab:baseline}.  This does not contradict \cref{thm:eprocess}, since the model-based predictive null is false when $c_{\rm mis}\ne1$.  It highlights the operational limitation of direct predictive-correction confirmation and motivates reference-calibrated variants.

\begin{table}[t]
\centering
\caption{Miscalibration sweep: data generated with variance $c_{\rm mis}\sigma^2$, 
e-process computed with model variance $\sigma^2$ and correction $c=1.8$.  
All rows share the same standardized residual draws.}
\label{tab:miscal}
\medskip
\begin{tabular}{lrrrrr}
\toprule
$c_{\rm mis}$ & 1.0 & 1.25 & 1.5 & 1.8 & 2.2\\
\midrule
Confirmation rate & 0.027 & 0.376 & 0.892 & 0.997 & 1.000\\
Mean $\log M_T$ & $-14.33$ & $-3.22$ & 7.89 & 21.22 & 39.00\\
\bottomrule
\end{tabular}
\end{table}

\subsection{Mixtures over correction magnitude}

For the true mean shift $\delta^*=0.45$, we form a uniform mixture over
\[
  \{-0.45,\,0.225,\,0.45,\,0.675,\,0.9\}.
\]
The mixture confirms in all replications with median stopping time $19$: slower than the oracle single correction (median $16$) but robust to uncertainty about the correction magnitude.  Under the predictive null, the same mixture confirms in probability $0.031$, verifying that the mixture e-process retains Type I control.

\subsection{Predictable plug-in corrections}
\label{sec:exp-predictable}

To exercise \cref{prop:predictable}, we replace the fixed $\delta$ by a predictable ridge estimate
\[
  \hat\delta_i
  =\mathrm{clip}\!\left(
  \frac{\sum_{j<i}g(X_j)\{Y_j-\mu_0(X_j)\}/\sigma^2}
       {1+\sum_{j<i}g^2(X_j)/\sigma^2},\;
  [-1.5,1.5]\right),
\]
computed from strictly past data, and bet with $\log e_i$ evaluated at $\hat\delta_i$.  Under the true shift $\delta^*=0.45$ the plug-in strategy confirms in all replications with median stopping time $25$ and mean final log wealth $32.03$---slower than the oracle ($16$, $34.65$) and comparable to the mixture ($19$, $33.14$), the price of learning the magnitude inside the wealth process.  Under the null it confirms in probability $0.025$, confirming validity.  As \cref{prop:predictable} notes, what is confirmed here is that an adaptive betting strategy found evidence against the source predictive null, not a single prespecified correction.

\subsection{Adaptive input selection}
\label{sec:exp-adaptive}

The conditional drift $\Gamma_h(x)$ of \cref{prop:drift} depends on the input, so an experimenter who controls the inputs can accelerate confirmation.  We implement an $\varepsilon$-greedy bandit ($\varepsilon=0.2$) over the input arms $\{-2,0,2\}$, with realized $\log e_i$ as the reward, so the design depends on past \emph{outcomes} and is genuinely adaptive.  Under the true shift $\delta^*=0.45$, the correctly specified per-arm drifts are $\KL(x)=\delta^2g^2(x)/(2\sigma^2)\in\{0.042,\,0.158,\,0.347\}$; the bandit concentrates on $x=2$ and achieves mean final log wealth $52.19$ versus $34.65$ under i.i.d.\ $\mathcal N(0,1)$ inputs (median stopping time $14$ versus $16$; the modest median gain reflects initialization and exploration overhead, while the $1.5\times$ drift gain compounds over the horizon).  Under the null the same adaptive design confirms in probability $0.031$: validity is unaffected by outcome-dependent input selection, exactly as \cref{thm:eprocess} asserts.

\begin{table}[t]
\centering
\caption{Robustness and adaptivity experiments ($T=200$, $5000$ replications).  Power rows use the true mean shift $\delta^*=0.45$; Type I rows use the predictive null.  The binomial Monte Carlo standard error of each reported confirmation rate is at most $0.0071$.}
\label{tab:robust}
\begin{tabular}{llrrr}
\toprule
Strategy & Condition & Confirm. & Median $\tau$ & Mean $\log M_T$\\
\midrule
Mixture over $\delta$ & power & 1.000 & 19 & 33.14\\
Mixture over $\delta$ & Type I & 0.031 & 18 & $-10.23$\\
Predictable plug-in $\hat\delta_i$ & power & 1.000 & 25 & 32.03\\
Predictable plug-in $\hat\delta_i$ & Type I & 0.025 & 16.5 & $-2.38$\\
Adaptive input selection & power & 0.998 & 14 & 52.19\\
Adaptive input selection & Type I & 0.031 & 14 & $-17.99$\\
\bottomrule
\end{tabular}
\end{table}

\subsection{Cross-family false confirmation}
\label{sec:exp-cross-family}

This experiment quantifies the limits established in \cref{sec:misspecification-control,sec:cross-family}: when the target predictive differs from the source, the moment condition \eqref{eq:moment-condition} is what preserves the same supermartingale proof, and both failure mechanisms of \cref{sec:moment-condition-failure} can occur at practically alarming rates when that condition breaks.  Results are in \cref{tab:crossfam} and \cref{fig:paths}b.

\begin{table}[t]
\centering
\caption{Cross-family confirmation ($T=200$, $5000$ replications).  $\overline\Gamma$ is the analytic per-step drift from \eqref{eq:var-drift-general} or \eqref{eq:mean-drift-general}.  No target in this table satisfies the moment condition \eqref{eq:moment-condition} for the tested correction: in the first, second, and fourth rows the relevant tilt moment is finite but strictly larger than the normalizer, while in the third and fifth rows (the two $t_5$ targets) it is $+\infty$.  In all rows the predictive null is false, so the source-null guarantee of \cref{thm:eprocess} does not apply.  The binomial Monte Carlo standard error of each reported confirmation rate is at most $0.0071$.}
\label{tab:crossfam}
\begin{tabular}{llrrrr}
\toprule
Target distribution & Tested correction & $\overline\Gamma$ & Confirm. & Median $\tau$ & Mean $\log M_T$\\
\midrule
Mean shift $\delta^*=0.45$ & variance $c=1.8$ & $+0.006$ & 0.619 & 55 & 1.16\\
Mean shift $\delta^*=0.6$ & variance $c=1.8$ & $+0.066$ & 0.972 & 36 & 13.16\\
$t_5$, matched variance & variance $c=1.8$ & $-0.072$ & 0.191 & 30 & $-14.24$\\
Variance inflation $1.8$ & mean $\delta=0.45$ & $-0.174$ & 0.141 & 13 & $-35.01$\\
$t_5$, matched variance & mean $\delta=0.45$ & $-0.174$ & 0.037 & 13.5 & $-34.80$\\
\bottomrule
\end{tabular}
\end{table}

\paragraph{Positive drift under the wrong family.}
A pure conditional mean shift with unchanged conditional variance inflates the second moment about $\mu_0$.  Because $\Gamma_c(x)$ in \eqref{eq:var-drift-general} is affine in $m_2(x)$, averaging the pointwise condition \eqref{eq:cross-family-threshold} over the input distribution gives positive average drift once $\delta^{*2}\E[g^2]/\sigma^2>0.3225$.  At $\delta^*=0.45$ the margin is thin ($\overline\Gamma=+0.006$), yet the confirmation rate is already $0.619$; at $\delta^*=0.6$ ($\overline\Gamma=+0.066$) it is $0.972$.  A practitioner who proposed a noise-inflation correction would confirm it with near certainty when the actual change is a response shift with no dispersion change at all.

\paragraph{Negative drift does not protect.}
Under a mean-preserving variance inflation ($c_{\rm mis}=1.8$, mean unchanged), the mean-correction e-process has strongly negative drift ($\overline\Gamma=-0.174$; mean final log wealth $-35.01$), yet it confirms in probability $0.141$, nearly three times the nominal level.  The mechanism is \cref{prop:misspecification-control}: the mean tilt is convex in $y$, so conditionally on the input $\E_q[e_i\mid X_i=x]=\exp\{(c_{\rm mis}-1)\delta^2g^2(x)/(2\sigma^2)\}>1$, breaking the supermartingale condition at every input.  Averaging over the input distribution---taking the expectation of the exponential, not the exponential of the expectation---predicts
\[
  \E_q[e_i]
  =\E_X\!\left[\exp\!\left\{\frac{(c_{\rm mis}-1)\delta^2g^2(X)}{2\sigma^2}\right\}\right]
  =1.1529,
\]
against an empirical mean of $1.154$; the corresponding Jensen lower bound $\exp\{(c_{\rm mis}-1)\delta^2\E[g^2]/(2\sigma^2)\}=1.1492$ is not the right prediction and understates the violation.  Under the heavy-tailed $t_5$ target with matched variance, the variance-correction e-process likewise has negative drift ($-0.072$) but confirms in probability $0.191$; here the mechanism is jumps rather than variance: $72\%$ of the crossings are produced by a \emph{single} observation whose quadratic log e-value exceeds the entire threshold (median crossing increment $4.9$ nats against a threshold of $3.0$).  The only row resembling nominal behavior is the mean correction under the $t_5$ target ($0.037$), and even that is not guaranteed by the present argument: the linear tilt has no moment generating function under a $t$ distribution, so $\E_q[h(x,Y)]=+\infty$, \eqref{eq:moment-condition} fails as badly as it can, and the rate merely happens to be small at this horizon.  This example emphasizes that failure of the conditional e-value moment condition can be severe even when the observed finite-horizon crossing rate happens to be small.

\paragraph{Interpretation.}
Confirmation is Neyman--Pearson evidence for $\phDobj$ against $\pzeroDobj$; it identifies neither the shift family nor the physical mechanism generating the evidence.  Where several prespecified mechanisms are scientifically plausible, the corresponding wealth paths can be inspected diagnostically, but a family-level confirmatory decision should use the panel mixture of \cref{prop:correction-panel} or an explicit error allocation.  The refutation boundary of \cref{sec:refutation} can additionally retire a wrongly proposed predictive correction when the source predictive accumulates sufficient relative evidence against it.
\subsection{Time-uniform Type I and overshoot accounting}
\label{sec:exp-overshoot}

Finally, we verify that the sub-$\alpha$ null confirmation rate is an overshoot effect, not a truncation effect.  Running the $\delta=0.45$ mean-correction e-process under the null for $5000$ replications to horizon $T=5000$ (an independent replication of the null condition), the cumulative confirmation rate is $0.0366$ at $t=200$ and \emph{identical} at $t=1000$ and $t=5000$: with null drift $-0.174$ per step, every crossing observed in this simulation occurs within the first few dozen observations, so the realized rate is not an artifact of stopping at $T=200$.

The equality \eqref{eq:overshoot-equality} applies here, since under $q=\phDobj$ the drift is $\KL=\delta^2\E[g^2]/(2\sigma^2)=0.1738>0$ and hence $\log M_t\to+\infty$ $P_h$-almost surely.  The mean wealth at crossing is $\E[M_{\tau^*}\mid\text{cross}]=27.7$, so the identity predicts $1/27.7=0.0361$, equivalently a product
\[
  \Prob(\text{cross})\cdot\E[M_{\tau^*}\mid\text{cross}]=0.0366\times27.7=1.014,
\]
against the theoretical value $1$.  There are $183$ crossing paths, and the binding uncertainty is the heavy-tailed conditional mean $\E[M_{\tau^*}\mid\text{cross}]$.  A separate $5000$-path null check (seed $12345$, horizon $T=200$) gives a Monte Carlo standard error of $0.95$ for this conditional mean and $0.083$ for the directly checked product $M_{\tau^*}\mathbf 1\{\tau^*<\infty\}$.  Thus the residual $1.4\%$ is well within Monte Carlo error; we do not claim agreement to a fixed number of digits.  For reference, the independent $T=200$ replication of \cref{tab:baseline} gives $0.0342$ for the same condition, so the run-to-run spread in the rate itself is of the same order as the discrepancy above.  The gap between the nominal $\alpha=0.05$ and the realized $\approx0.036$ is therefore accounted for by the discrete overshoot $M_{\tau^*}>1/\alpha$ in this experiment, and would shrink only if the per-step evidence increments were made smaller.

\section{Discussion and Conclusion}
\label{sec:discussion}
\label{sec:conclusion}

We developed an anytime-valid framework for evaluating a prespecified predictive
correction from sequentially observed target outcomes. Conditional on the
realized training data, a nonnegative tilt transforms the source predictive
distribution into a corrected predictive distribution, and the resulting
corrected-to-source predictive ratio yields a conditional e-value. Its running
product forms a nonnegative martingale under the source predictive null, so the
correction can be monitored continuously and evaluated at data-dependent
stopping times without inflating the probability of false confirmation. The
logarithm of this wealth process is the cumulative predictive log-score
advantage of the corrected predictive over the source predictive. A boundary
crossing therefore provides anytime-valid \emph{relative confirmation}: it
supports replacing the source predictive by the proposed correction, but does
not imply that the corrected predictive is the true target predictive or that
the mechanism encoded by the correction uniquely explains the observed shift.

The drift analysis clarifies when such evidence should accumulate. For a target
conditional distribution \(q\),
\[
\Gamma_h(x)
=
\KL\!\left(
q(\cdot\mid x)\,\middle\|\,\pzeroD{\cdot}{x}
\right)
-
\KL\!\left(
q(\cdot\mid x)\,\middle\|\,\phD{\cdot}{x}
\right)
\]
whenever the relevant divergences are finite. Thus, positive drift means that
the corrected predictive is closer to the target than the source predictive in
conditional Kullback--Leibler divergence. Exact specification is sufficient
but not necessary. Under stable sampling, the average drift gives the
asymptotic evidence gained per observation, while the finite-horizon bounds
translate this growth rate into explicit control of delayed confirmation.
Validity can also persist beyond the source predictive null: targets satisfying
\[
\E_{Y\sim q_i(\cdot\mid X_i)}
\bigl[h(X_i,Y)\bigr]
\le
\ZhD{X_i}
\]
keep the evidence process supermartingale-like and therefore preserve the same
time-uniform false-confirmation bound. This protected region is a robustness
property of the directional betting strategy, not an enlargement of the
scientific null. Conversely, negative long-run drift alone does not imply
time-uniform protection, since early variability or rare large jumps may still
produce a boundary crossing.

The same construction accommodates a range of structured predictive
corrections. Label tilts represent label-shift corrections, Gaussian linear and
quadratic tilts yield conditional mean and variance corrections, subgroup
interactions permit localized changes, and exponential-family tilts provide a
general representation for prespecified feature directions. Mixtures allow
uncertainty over a collection of corrections to be incorporated without
choosing one component before monitoring, whereas predictable tilts allow the
betting strategy to adapt to past observations and the current input.
Beyond-tolerance comparisons address a different operational question by
testing whether the shift is large enough to favor an actionable correction
over an entire tolerated region rather than merely detecting any departure
from the source predictive.

Exact normalization provides additional structure. When the source and
corrected predictives are mutually absolutely continuous, the reciprocal
likelihood ratio yields an anytime-valid lower boundary for refuting the
corrected predictive in favor of the source predictive. The two boundaries
correspond to distinct rejection guarantees under different predictive nulls;
neither establishes that one of the two predictives is the true target
distribution. The overshoot identity explains why the realized source-null
crossing probability can be strictly below the nominal level. A certified
upper bound on the normalizer still preserves conservative upper-bound
validity, but generally sacrifices the exact log-score interpretation and the
reciprocal refutation guarantee.

The synthetic experiments support these theoretical conclusions in controlled
settings. Correctly specified mean and variance corrections accumulate evidence
at their predicted rates, moderate mismatch can slow evidence growth without
eliminating it, and severe mismatch can reverse the drift. Adaptive input
selection can accelerate evidence accumulation without compromising
source-null validity, while cross-family experiments illustrate the principal
interpretive limitation: a correction-specific evidence process can respond to
changes generated by a different mechanism. The evidence therefore concerns
predictive advantage relative to the source reference rather than unique
mechanistic identification.

Several limitations remain. The guarantees are conditional on the fitted source
predictive and therefore do not automatically account for source-model
misspecification or uncertainty introduced during model fitting. The reciprocal
refutation guarantee developed here is calibrated under the corrected
predictive null, and we do not characterize a broader misspecification class
for the lower boundary. The correction, mixture weights, tolerance boundary,
and monitoring rule must be prespecified or chosen predictably under the stated
filtration, and useful evidence accumulation requires inputs that are
informative for distinguishing the source and corrected predictives. Pure
covariate shift is outside the present conditional-outcome framework and
requires separate monitoring of the input distribution.

Natural extensions include reference-calibrated or conformal layers that
protect against source-predictive misspecification, experimental-design
procedures that select informative inputs while preserving anytime validity,
family-level methods for principled post-confirmation selection among competing
corrections, and evaluation with fitted predictive models and
application-driven corrections on real-world data. Overall, the framework
provides a direct path from a scientifically motivated predictive correction to
continuously monitored, finite-sample-valid relative evidence.

\bibliographystyle{apalike}
\bibliography{sjc}

\clearpage
\appendix
\section{Proofs of Main Results and Additional Derivations}
\label{app:proofs}

\subsection{Proof of Lemma~\ref{lem:ratio}: Normalized Tilt as a Likelihood Ratio}
\label{app:proof-ratio}
For each fixed $x$, nonnegativity of $h$ and \eqref{eq:normalizer} imply
\[
  \int \phD{y}{x}\,dy
  =\frac{1}{\ZhD{x}}\int h(x,y)\pzeroD{y}{x}\,dy
  =1.
\]
Thus $\phD{\cdot}{x}$ is a probability distribution.  It is absolutely continuous with respect to $\pzeroD{\cdot}{x}$ because its density is obtained by multiplying the source density by the nonnegative factor $h(x,\cdot)/\ZhD{x}$.  The Radon--Nikodym ratio is therefore $h(x,y)/\ZhD{x}$ $\pzeroD{\cdot}{x}$-a.s., and integrating this ratio under $\pzeroD{\cdot}{x}$ gives one.

\subsection{Proof of Proposition~\ref{prop:evalue}: Per-Observation Conditional E-Value}
\label{app:proof-evalue}
Condition on $\cG_i=\sigma(\cF_{i-1},X_i)$.  Under $H_0^{\mathrm{pred}}$, the conditional distribution of $Y_i$ is $\pzeroD{\cdot}{X_i}$, while $X_i$ and $\ZhD{X_i}$ are fixed.  Hence
\[
  \E[e_i\mid\cG_i]
  =\frac{1}{\ZhD{X_i}}\int h(X_i,y)\pzeroD{y}{X_i}\,dy
  =1.
\]
This is exactly the conditional e-value property.

\subsection{Proof of Theorem~\ref{thm:eprocess}: Anytime-Valid Predictive-Correction Confirmation}
\label{app:proof-eprocess}
By \cref{prop:evalue},
\[
  \E[e_t\mid\cF_{t-1}]
  =\E\!\left[\E[e_t\mid\cG_t]\mid\cF_{t-1}\right]
  =1.
\]
Since $M_{t-1}$ is $\cF_{t-1}$-measurable,
\[
  \E[M_t\mid\cF_{t-1}]
  =M_{t-1}\E[e_t\mid\cF_{t-1}]
  =M_{t-1}.
\]
Thus $(M_t)$ is a nonnegative martingale with $M_0=1$.  Ville's inequality gives
\[
  \sup_{P\in\mathcal P_0^{\mathrm{pred}}}
  \Prob_P\!\left(\sup_{t\ge0}M_t>\frac1\alpha\right)\le\alpha,
\]
and the stopping-time statement follows because $\{\tau^*<\infty\}=\{\sup_tM_t>1/\alpha\}$.

\subsection{Additional Details on Safe Numerical Approximation of the Normalizer}
\label{app:normalizer-numerics}
Let $\widetilde Z_i(X_i)$ be positive, $\cG_i$-measurable, and satisfy $\widetilde Z_i(X_i)\ge\ZhD{X_i}$ almost surely.  Then
\[
  \E\!\left[
    \frac{h(X_i,Y_i)}{\widetilde Z_i(X_i)}
    \,\middle|\,
    \cG_i
  \right]
  =\frac{\ZhD{X_i}}{\widetilde Z_i(X_i)}
  \le1,
\]
so sequential composition yields a nonnegative supermartingale.  Relative to exact normalization, the log increment is reduced by
\[
  \log\frac{\widetilde Z_i(X_i)}{\ZhD{X_i}}.
\]

For a random numerical estimate, two conditioning arguments show the problem.  First, once a positive estimate $\widehat Z_i$ is generated before $Y_i$ and included in the pre-outcome information, the conditional mean of the approximate factor is $\ZhD{X_i}/\widehat Z_i$ and exceeds one on every undershoot.  Second, suppose instead that the auxiliary randomness is averaged out, is conditionally independent of $Y_i$ given $\cG_i$, and satisfies $\E[\widehat Z_i\mid\cG_i]=\ZhD{X_i}$.  Then
\[
  \E\!\left[
    \frac{h(X_i,Y_i)}{\widehat Z_i}
    \,\middle|\,
    \cG_i
  \right]
  =\ZhD{X_i}\,
    \E\!\left[\frac1{\widehat Z_i}\,\middle|\,\cG_i\right]
  \ge1
\]
by Jensen's inequality, with strict inequality unless the estimate is exact almost surely.  Thus unbiasedness of $\widehat Z_i$ does not imply validity after inversion.

Finally, replacing $\ZhD{X_i}$ by an upper bound changes the reciprocal factor from $\ZhD{X_i}/h(X_i,Y_i)$ to $\widetilde Z_i(X_i)/h(X_i,Y_i)$.  Under the strict-positivity condition imposed in \cref{sec:two-boundary}, $H_h^{\mathrm{pred}}$ gives the latter conditional mean $\widetilde Z_i(X_i)/\ZhD{X_i}\ge1$, so it is not generally an e-value for refutation.  The lower boundary in \cref{sec:refutation} therefore requires exact normalization.

\subsection{Proof of Proposition~\ref{prop:drift}: Conditional Drift Decomposition}
\label{app:proof-drift}
The identity $\E_q[\log e_i\mid\cG_i]=\Gamma_h(X_i)$ follows from $Y_i\mid\cG_i\sim q(\cdot\mid X_i)$ and the definition of $e_i$.  The hypothesis $\E_q|\log e_i|<\infty$ gives
\[
  \E_q|\Gamma_h(X_i)|
  =\E_q\big|\E_q[\log e_i\mid\cG_i]\big|
  \le\E_q|\log e_i|<\infty
\]
by conditional Jensen.  Hence $\log M_t$, $\sum_{i\le t}\Gamma_h(X_i)$, and $N_t$ are integrable.  When the two KL terms are finite, adding and subtracting $\log q(y\mid x)$ gives
\begin{align*}
  \Gamma_h(x)
  &=\int q(y\mid x)\log\frac{q(y\mid x)}{\pzeroD{y}{x}}\,dy
    -\int q(y\mid x)\log\frac{q(y\mid x)}{\phD{y}{x}}\,dy\\
  &=\KL\{q(\cdot\mid x)\|\pzeroD{\cdot}{x}\}
    -\KL\{q(\cdot\mid x)\|\phD{\cdot}{x}\}.
\end{align*}
The increments $D_i=\log e_i-\Gamma_h(X_i)$ satisfy $\E_q[D_i\mid\cG_i]=0$ and therefore also $\E_q[D_i\mid\cF_{i-1}]=0$, so $N_t=\sum_{i\le t}D_i$ is a martingale under $P_q$.  Under $\sup_i\E[D_i^2\mid\cG_i]\le v$,
\[
  \sum_{i=1}^\infty\frac{\E_q[D_i^2]}{i^2}
  \le v\sum_{i=1}^\infty i^{-2}<\infty.
\]
The martingale strong law, equivalently Chow's theorem followed by Kronecker's lemma, yields $N_t/t\to0$ almost surely.  On the event
\[
  \liminf_{t\to\infty}\frac1t\sum_{i=1}^t\Gamma_h(X_i)>0,
\]
choose $\epsilon>0$ smaller than half this liminf.  Then $N_t/t\ge-\epsilon$ eventually and the average drift is at least $2\epsilon$ eventually, so $\log M_t\ge\epsilon t$ eventually.  Hence $\log M_t\to\infty$ and $\tau^*<\infty$.

\subsection{Proof of Corollary~\ref{cor:ergodic}: Asymptotic Growth Under I.I.D. or Stationary-Ergodic Sampling}
\label{app:proof-ergodic}
In the i.i.d.\ case, the sequential assumptions imply that the pairs $(X_i,Y_i)$ are i.i.d.\ with joint distribution $q^X(dx)q(dy\mid x)$.  In the more general case, stationarity and ergodicity of the pair process are assumed directly.  Since $\log e_i$ is a fixed measurable function of $(X_i,Y_i)$ and $\E_{q^Xq}|\log e_1|<\infty$, the ordinary strong law, respectively Birkhoff's theorem, gives
\[
  \frac1t\log M_t
  =\frac1t\sum_{i=1}^t\log e_i
  \longrightarrow \E_{q^X q^{Y\mid X}}[\log e_1]
  =\E_{q^X}[\Gamma_h(X)]
  \quad\text{a.s.}
\]

\subsection{Proof of Corollary~\ref{cor:finite-horizon}: Finite-Horizon Crossing Bounds}
\label{app:proof-finite-horizon}
Write $\gamma=\overline\Gamma(q;h)>0$, $S_t=\log M_t=\sum_{i=1}^t Z_i$, and $a_t=t\gamma-b>0$.  Since $\{\tau^*>t\}\subseteq\{S_t\le b\}$,
\[
  \Prob_q(\tau^*>t)
  \le
  \Prob_q\{S_t-t\gamma\le-a_t\}.
\]
Independence and $\Var_q(Z_i)\le v$ imply $\Var_q(S_t)\le tv$.  Cantelli's one-sided inequality therefore yields
\[
  \Prob_q(\tau^*>t)
  \le
  \frac{\Var_q(S_t)}{\Var_q(S_t)+a_t^2}
  \le
  \frac{tv}{tv+a_t^2},
\]
which is \eqref{eq:finite-horizon-power}.

If each centered increment is sub-Gaussian with variance proxy $s^2$, independence gives
\[
  \E_q\exp\{\theta(S_t-t\gamma)\}
  \le
  \exp\!\left(\frac{t\theta^2s^2}{2}\right).
\]
Applying the Chernoff bound to $-(S_t-t\gamma)$ and optimizing at $\theta=a_t/(ts^2)$ gives
\[
  \Prob_q(S_t-t\gamma\le-a_t)
  \le
  \exp\!\left(-\frac{a_t^2}{2ts^2}\right),
\]
which is \eqref{eq:finite-horizon-subgaussian}.

Finally, if $|Z_i-\gamma|\le R$ almost surely, the one-sided Bernstein inequality for independent centered increments with total variance at most $tv$ gives
\[
  \Prob_q(S_t-t\gamma\le-a_t)
  \le
  \exp\!\left(
    -\frac{a_t^2}{2tv+\tfrac23Ra_t}
  \right),
\]
which is \eqref{eq:finite-horizon-bernstein}.

\subsection{Interpreting the Finite-Horizon Crossing Bounds}
\label{app:finite-horizon-guide}

This subsection gives an elementary interpretation of
\cref{cor:finite-horizon}.  It does not introduce a new result; its purpose is
to explain what the three bounds say, why the quantity
$t\overline\Gamma(q;h)-b$ appears, and how the bounds should be used.

\paragraph{Accumulated evidence and the confirmation boundary.}
Write
\[
  \gamma=\overline\Gamma(q;h)>0,
  \qquad
  S_t=\log M_t=\sum_{i=1}^t Z_i,
  \qquad
  b=\log(1/\alpha).
\]
The correction is confirmed at
\[
  \tau^*=\inf\{t\ge1:S_t>b\}.
\]
Thus $S_t$ is the accumulated log evidence and $b$ is the amount of log
evidence required for confirmation.  For example, when $\alpha=0.05$,
\[
  b=\log 20\approx3.
\]
Under the assumptions of \cref{cor:finite-horizon},
\[
  \E_q[S_t]=t\gamma.
\]
Ignoring random fluctuation, the accumulated evidence reaches the boundary
when $t\gamma\approx b$.  This gives the first-order crossing-time heuristic
\begin{equation}
  \tau^*
  \approx
  \frac{b}{\gamma}
  =
  \frac{\log(1/\alpha)}{\overline\Gamma(q;h)}.
  \label{eq:finite-horizon-heuristic-guide}
\end{equation}
A larger average log-score advantage $\gamma$ therefore means faster expected
confirmation, while a more stringent level $\alpha$ raises the boundary and
requires more observations.

\paragraph{Why delayed confirmation is a lower-tail event.}
The event $\{\tau^*>t\}$ means that the process has not crossed the boundary
at any time up to $t$.  In particular, its endpoint must satisfy $S_t\le b$.
Consequently,
\begin{align}
  \{\tau^*>t\}
  &\subseteq
  \{S_t\le b\} \notag\\
  &=
  \left\{
    S_t-t\gamma
    \le
    -\{t\gamma-b\}
  \right\}.
  \label{eq:finite-horizon-delay-event}
\end{align}
When $t\gamma>b$, the mean accumulated evidence is already above the
boundary.  Failure to confirm by time $t$ then requires a downward fluctuation
of at least $t\gamma-b$.  The three inequalities in
\cref{cor:finite-horizon} are simply three ways to bound the probability of
this unfavorable fluctuation.

The inclusion in \eqref{eq:finite-horizon-delay-event} is one-way.  A path may
cross before time $t$ and later return below $b$, in which case $S_t\le b$ but
$\tau^*\le t$.  The concentration bounds may therefore be conservative even
before accounting for looseness in the concentration inequality itself.

\paragraph{Cantelli bound.}
If only the one-step variance bound
$\Var_q(Z_i)\le v$ is available, then
$\Var_q(S_t)\le tv$.  Cantelli's one-sided inequality gives
\[
  \Prob_q(\tau^*>t)
  \le
  \frac{tv}{tv+(t\gamma-b)^2}.
\]
The numerator $tv$ measures accumulated noise, while $(t\gamma-b)^2$ is the
squared evidence margin above the boundary.  In schematic form,
\[
  \text{delayed-confirmation probability}
  \ \lesssim\
  \frac{\text{noise}}
       {\text{noise}+\text{squared signal margin}}.
\]
For large $t$, the bound behaves approximately as
\[
  \frac{v}{t\gamma^2},
\]
so it decreases at the polynomial rate $1/t$.  Its advantage is that it
requires only a finite variance bound.

\paragraph{Sub-Gaussian bound.}
If the centered increments have sub-Gaussian variance proxy $s^2$, then
\[
  \Prob_q(\tau^*>t)
  \le
  \exp\left\{
    -\frac{(t\gamma-b)^2}{2ts^2}
  \right\}.
\]
The same squared evidence margin appears in the numerator, but stronger
tail control yields an exponential bound.  For large $t$,
\[
  \exp\left\{
    -\frac{(t\gamma-b)^2}{2ts^2}
  \right\}
  \approx
  \exp\left\{
    -\frac{t\gamma^2}{2s^2}
  \right\}.
\]
Hence the probability of delayed confirmation can decrease exponentially
rather than at the $1/t$ rate.  The practical usefulness of this bound
depends on the quality of the certified proxy $s^2$: a very loose proxy can
make the exponential bound numerically weak at moderate horizons.

\paragraph{Bernstein bound.}
If the centered increments are bounded by $R$ and have variance at most $v$,
then
\[
  \Prob_q(\tau^*>t)
  \le
  \exp\left\{
    -\frac{(t\gamma-b)^2}
    {2tv+\tfrac23R(t\gamma-b)}
  \right\}.
\]
This bound uses both the typical scale of fluctuation, represented by $v$,
and the largest possible fluctuation, represented by $R$.  It can improve on
a range-based sub-Gaussian bound when the variance is much smaller than the
worst-case range would suggest.  The Cantelli and Bernstein bounds need not
be ordered uniformly: Cantelli can be sharper near the nominal crossing
time, while Bernstein can become substantially sharper at longer horizons.

\paragraph{From a delayed-crossing bound to finite-horizon power.}
Each displayed inequality has the form
\[
  \Prob_q(\tau^*>t)\le B_t.
\]
Equivalently,
\begin{equation}
  \Prob_q(\tau^*\le t)\ge1-B_t.
  \label{eq:finite-horizon-power-guide}
\end{equation}
Thus the corollary provides a conservative lower bound on the probability
that the correction has been confirmed by time $t$.  It does not give the
exact distribution of $\tau^*$, and the bounds are informative only after
the expected accumulated evidence exceeds the boundary, that is, after
$t\gamma>b$.

\paragraph{A numerical illustration.}
Take $\alpha=0.05$, so $b=\log20\approx3$, and suppose
\[
  \gamma=0.1,
  \qquad
  v=0.2,
  \qquad
  s^2=0.2.
\]
The heuristic \eqref{eq:finite-horizon-heuristic-guide} gives
\[
  \tau^*\approx\frac{3}{0.1}=30.
\]
At $t=50$, the expected accumulated log evidence is $5$, only about $2$
above the boundary.  The Cantelli bound is
\[
  \Prob_q(\tau^*>50)
  \le
  \frac{50(0.2)}{50(0.2)+2^2}
  =
  \frac{10}{14}
  \approx0.714,
\]
while the sub-Gaussian bound is
\[
  \Prob_q(\tau^*>50)
  \le
  \exp\left\{-\frac{2^2}{2(50)(0.2)}\right\}
  =
  e^{-0.2}
  \approx0.819.
\]
The exponential bound is not automatically sharper at a short horizon,
especially when its variance proxy is conservative.

At $t=200$, the evidence margin is $20-3=17$.  The two bounds become
\[
  \Prob_q(\tau^*>200)
  \le
  \frac{40}{40+17^2}
  \approx0.122
\]
and
\[
  \Prob_q(\tau^*>200)
  \le
  \exp\left\{-\frac{17^2}{2(200)(0.2)}\right\}
  \approx0.027.
\]
The latter implies
\[
  \Prob_q(\tau^*\le200)\ge0.973.
\]
This example illustrates the basic message: positive mean log evidence
determines the approximate crossing time, while concentration controls how
likely random fluctuation is to delay confirmation beyond a chosen horizon.

\paragraph{Summary.}
The logical chain is
\[
  \begin{gathered}
    \text{positive mean log evidence}\\
    \Downarrow\\
    \text{expected evidence reaches the boundary near }b/\gamma\\
    \Downarrow\\
    \text{concentration bounds the probability of a delayed crossing}.
  \end{gathered}
\]
Finite variance yields a broadly applicable polynomial bound, sub-Gaussian
tails yield an exponential bound, and bounded increments together with a
variance bound yield the Bernstein alternative.  The corollary therefore
strengthens the asymptotic statement of eventual confirmation into an
explicit finite-horizon guarantee.

\subsection{Proof of Corollary~\ref{cor:correct}: Correctly Specified Predictive Correction}
\label{app:proof-correct}
If $q(\cdot\mid x)=\phD{\cdot}{x}$, then the KL decomposition in \eqref{eq:gamma-kl} gives
\begin{align*}
  \Gamma_h(x)
  &=\KL\{\phD{\cdot}{x}\|\pzeroD{\cdot}{x}\}
    -\KL\{\phD{\cdot}{x}\|\phD{\cdot}{x}\}\\
  &=\KL\{\phD{\cdot}{x}\|\pzeroD{\cdot}{x}\}\ge0.
\end{align*}
The adaptive-input conclusion follows from \cref{prop:drift}. Under either
sampling regime of \cref{cor:ergodic}, that corollary yields
$t^{-1}\log M_t\to\gamma_h$ almost surely, and $\gamma_h>0$ implies eventual
crossing. Finally, the conditional KL divergence is nonnegative, so
$\gamma_h=0$ if and only if
$\KL\!\left(\phD{\cdot}{X}\,\middle\|\,\pzeroD{\cdot}{X}\right)=0$ for
$q^X$-almost every $X$, which is equivalent to
$\phD{\cdot}{X}=\pzeroD{\cdot}{X}$ there.

\subsection{Proof of Proposition~\ref{prop:predictable}: Predictable Tilts}
\label{app:proof-predictable}
Condition on $\cG_i$.  By assumption, $h_i(X_i,\cdot)$ and $Z_i(X_i)$ are then fixed functions of the yet-unobserved outcome, while $Y_i\sim\pzeroD{\cdot}{X_i}$ under the null.  Therefore
\[
  \E[e_i\mid\cG_i]
  =\frac{1}{Z_i(X_i)}\int h_i(X_i,y)\pzeroD{y}{X_i}\,dy
  =1.
\]
Iterating conditional expectations and multiplying sequentially gives the e-process property exactly as in \cref{app:proof-eprocess}.

\subsection{Proof of Proposition~\ref{prop:misspecification-control}: False-Confirmation Control under Target Misspecification}
\label{app:proof-misspecification-control}
Under $Y_i\mid\cG_i\sim q_i(\cdot\mid X_i)$,
\[
  \E[e_i\mid\cG_i]
  =\frac{\E_{Y\sim q_i(\cdot\mid X_i)}[h(X_i,Y)]}{\ZhD{X_i}}.
\]
If \eqref{eq:moment-condition} holds at $X_i$ almost surely for every $i$, then this conditional expectation is at most one.  Consequently,
\[
  \E[M_t\mid\cF_{t-1}]
  =M_{t-1}\E[e_t\mid\cF_{t-1}]
  \le M_{t-1},
\]
where the inequality follows by conditioning first on $\cG_t$ and then on $\cF_{t-1}$.  Thus $(M_t)$ is a nonnegative supermartingale and Ville's inequality gives the $\alpha$ bound.  If the pointwise condition holds for every $x$, the argument applies to every input process.  Failure of the condition only removes this supermartingale proof; it does not imply a converse, as explained in \cref{sec:moment-condition-failure}.

\subsection{Geometry of the Protected Half-Space}
\label{app:halfspace-details}
Fix $x$ and abbreviate $P_0=\pzeroD{\cdot}{x}$, $Z=\ZhD{x}$, and
$H(Y)=h(x,Y)$.  Suppose $H$ is not $P_0$-almost surely constant.  If
$P_0(H>Z)=0$, then $H\le Z$ almost surely; nonconstancy then forces
$P_0(H<Z)>0$, whence $\E_{P_0}H<Z$, contradicting $\E_{P_0}H=Z$.
Hence the set $A=\{H>Z\}$ has positive $P_0$-probability.  Let $R=P_0(\,\cdot\mid A)$ and $Q_\varepsilon=(1-\varepsilon)P_0+\varepsilon R$.  Then $R\ll P_0$ and $\E_RH>Z$, so
\[
  \E_{Q_\varepsilon}H
  =(1-\varepsilon)Z+\varepsilon\E_RH>Z
\]
for every $\varepsilon>0$.  Thus $Q_\varepsilon\notin\mathcal H_h(x)$.  Its density relative to $P_0$ is
\[
  \frac{dQ_\varepsilon}{dP_0}
  =1-\varepsilon+\varepsilon\frac{\mathbf1_A}{P_0(A)},
\]
which converges uniformly to one as $\varepsilon\downarrow0$.  Hence total variation, Hellinger distance, and $\KL(Q_\varepsilon\|P_0)$ all converge to zero.  This proves the neighborhood claim.

At the corrected predictive,
\[
  \E_{\phD{\cdot}{x}}\!\left[\frac{H}{Z}\right]
  =\int\frac{H}{Z}\frac{H}{Z}\,dP_0
  =\E_{P_0}\!\left[\left(\frac{H}{Z}\right)^2\right]
  =1+\Var_{P_0}\!\left(\frac{H}{Z}\right),
\]
with the final quantity interpreted in the extended sense.  It exceeds one for every nontrivial tilt.  Finally, if $Q\in\mathcal H_h(x)$, then $\E_Q[e]\le1$, and Jensen gives
\[
  \E_Q[\log e]\le\log\E_Q[e]\le0
\]
whenever the logarithmic expectation is well defined.

\subsection{Proof of Proposition~\ref{prop:overshoot}: Overshoot Identity}
\label{app:proof-overshoot}

\paragraph{The inequality.}
$(M_{t\wedge\tau^*})$ is a nonnegative $P_0$-martingale with $\E_{P_0}[M_{t\wedge\tau^*}]=1$.  As $t\to\infty$,
\[
  M_{t\wedge\tau^*}
  \longrightarrow
  M_{\tau^*}\mathbf1\{\tau^*<\infty\}
  +M_\infty\mathbf1\{\tau^*=\infty\}
  \quad\text{a.s.}
\]
The limit $M_\infty$ exists by nonnegative martingale convergence.  Fatou's lemma gives
\[
  \E_{P_0}[M_{\tau^*}\mathbf1\{\tau^*<\infty\}]
  \le1.
\]
Since $M_{\tau^*}>1/\alpha$ on $\{\tau^*<\infty\}$, factor the left side as
\[
  P_0(\tau^*<\infty)\,
  \E_{P_0}[M_{\tau^*}\mid\tau^*<\infty]
\]
to obtain \eqref{eq:overshoot-ineq}.

\paragraph{The exact identity.}
Let $P_0$ and $P_h$ denote the distributions of the data stream under
$H_0^{\mathrm{pred}}$ and $H_h^{\mathrm{pred}}$, with the same conditional
input mechanism under both. By \cref{lem:ratio}, $P_h\ll P_0$ on each
$\cF_t$ even without strict positivity of $h$; the input factors then cancel,
and
\[
  \frac{dP_h|_{\cF_t}}{dP_0|_{\cF_t}}
  =\prod_{i=1}^t
    \frac{\phD{Y_i}{X_i}}{\pzeroD{Y_i}{X_i}}
  =M_t.
\]
For every finite $t$, $\{\tau^*=t\}\in\cF_t$, so
\[
  P_h(\tau^*=t)
  =\E_{P_0}[M_t\mathbf1\{\tau^*=t\}]
  =\E_{P_0}[M_{\tau^*}\mathbf1\{\tau^*=t\}].
\]
Summing over $t\ge1$ and using monotone convergence gives
\[
  \E_{P_0}[M_{\tau^*}\mathbf1\{\tau^*<\infty\}]
  =P_h(\tau^*<\infty),
\]
which is \eqref{eq:overshoot-exact}; factoring the left side gives
\eqref{eq:overshoot-ratio} whenever $P_0(\tau^*<\infty)>0$.  If
$\log M_t\to+\infty$ $P_h$-almost surely, then
$P_h(\tau^*<\infty)=1$.  This also forces
$P_0(\tau^*<\infty)>0$: otherwise $P_0(\tau^*=t)=0$ for every finite $t$, and
the finite-time change-of-measure identity would give
\[
  P_h(\tau^*=t)
  =
  \E_{P_0}[M_t\mathbf1\{\tau^*=t\}]
  =0
  \qquad\text{for every }t,
\]
contradicting $P_h(\tau^*<\infty)=1$.  Therefore
\eqref{eq:overshoot-equality} follows.

\paragraph{On the uniform-integrability route.}
If $M_\infty=0$ almost surely, uniform integrability of $(M_{t\wedge\tau^*})$ is equivalent to $L^1$ convergence to $M_{\tau^*}\mathbf1\{\tau^*<\infty\}$ and hence to preservation of the expectation at the limit.  It is therefore equivalent to the desired equality rather than an independently checkable sufficient condition.  The change-of-measure argument avoids this circularity.

\subsection{Proof of Proposition~\ref{prop:tolerance}: Composite Tolerance Null}
\label{app:proof-tolerance}
Condition on $\cG_i$ and write $x=X_i$ and $\eta=\eta_i$.  Under $Y_i\sim p_\eta(\cdot\mid x,\Dtr)$,
\[
  \E[\exp\{\Delta\phi(x,Y_i)\}\mid\cG_i]
  =\frac{Z_{\eta+\Delta}(x)}{Z_\eta(x)}.
\]
Therefore
\begin{align*}
  \E[e_i^{\rm tol}\mid\cG_i]
  &=
  \frac{Z_{\eta+\Delta}(x)}{Z_\eta(x)}
  \frac{Z_{\eta_{\rm tol}}(x)}
       {Z_{\eta_{\rm tol}+\Delta}(x)}\\
  &=
  \exp\!\left(
  [\psi_x(\eta+\Delta)-\psi_x(\eta)]
  -[\psi_x(\eta_{\rm tol}+\Delta)-\psi_x(\eta_{\rm tol})]
  \right).
\end{align*}
Because $\psi_x$ is convex, the increment map
\[
  \eta\longmapsto\psi_x(\eta+\Delta)-\psi_x(\eta)
\]
is nondecreasing wherever both endpoints lie in the common interval $\mathcal I$.  Here $\eta+\Delta\in\mathcal I$ follows automatically because $\mathcal I$ is an interval containing $\eta$ and $\eta_{\rm tol}+\Delta$, with $\eta\le\eta+\Delta\le\eta_{\rm tol}+\Delta$.  Thus the exponent is nonpositive when $\eta\le\eta_{\rm tol}$, proving $\E[e_i^{\rm tol}\mid\cG_i]\le1$.  Sequential composition gives the supermartingale and Ville bounds.  Strict convexity is needed only for a strict or converse implication.

\subsection{Proof of Proposition~\ref{prop:correction-panel}: Mixture and Correction Panel}
\label{app:proof-mixture}
For each $\theta$, $(M_t(\theta))$ is a nonnegative martingale with $M_0(\theta)=1$ under the null.  Assume that $(\theta,\omega)\mapsto M_t(\theta)(\omega)$ is jointly measurable with respect to $\mathcal B(\Theta)\otimes\cF_t$, which follows, for example, when $(\theta,x,y)\mapsto h_\theta(x,y)$ is jointly measurable.  Conditional Tonelli then gives
\[
  \E\big[M_t^{\mathrm{mix}}\,\big|\,\cF_{t-1}\big]
  =\int\E[M_t(\theta)\mid\cF_{t-1}]\,d\Pi(\theta)
  =\int M_{t-1}(\theta)\,d\Pi(\theta)
  =M_{t-1}^{\mathrm{mix}},
\]
provided $\Pi$ is fixed before testing.  Hence $M_t^{\mathrm{mix}}$ is a nonnegative martingale.

\subsection{Gaussian Exponential-Tilt Derivation for Section~\ref{sec:exp-tilt}}
\label{app:gaussian-tilt}
Let $\pzeroD{y}{x}=\mathcal N\!\left( y \mid \mu_0(x),\sigma_0^2(x) \right)$ and choose $h_\eta(x,y)=\exp\{\eta g(x)y\}$.  The normalizer is
\[
  Z_\eta(x)=\exp\left\{\eta g(x)\mu_0(x)+\frac{1}{2}\eta^2g^2(x)\sigma_0^2(x)\right\}.
\]
The corrected predictive is Gaussian with mean
\[
  \mu_0(x)+\eta g(x)\sigma_0^2(x)
\]
and unchanged variance $\sigma_0^2(x)$.  Thus an exponential tilt in $g(x)y$ is equivalent to a conditional mean correction whose size scales with the source predictive variance; with $g\equiv1$ and $\eta=\beta$ this is the label-shift equivalence used in the sanity check.

The raw feature $g(x)y$ and the standardized-residual feature
$\phi_{\rm mean}(x,y)=g(x)\{y-\mu_0(x)\}/\sigma_0^2(x)$ are related, but the
parameter mapping must respect heteroscedasticity.  Indeed
\[
  \exp\!\left\{
    \frac{\delta g(x)\{y-\mu_0(x)\}}{\sigma_0^2(x)}
  \right\}
  =
  \exp\!\left\{
    -\frac{\delta g(x)\mu_0(x)}{\sigma_0^2(x)}
  \right\}
  \exp\!\left\{
    \frac{\delta g(x)y}{\sigma_0^2(x)}
  \right\},
\]
and the first factor depends only on $x$, so it cancels in the normalization
\eqref{eq:tilted-general}.  Thus the standardized feature is equivalent to a
raw linear tilt with input-dependent coefficient
$\eta(x)=\delta/\sigma_0^2(x)$.  If the source variance is homoscedastic,
$\sigma_0^2(x)\equiv\sigma_0^2$, this reduces to the scalar relation
$\delta=\eta\sigma_0^2$.  With a scalar $\eta$ and heteroscedastic
$\sigma_0^2(x)$, however, the raw feature $g(x)y$ produces the different mean
correction $\eta g(x)\sigma_0^2(x)$.  The standardized form is used in the main
text because its scalar parameter $\eta=\delta$ is directly the additive shift
multiplier in $\mu_0(x)+\delta g(x)$, which is the parametrization used in
\cref{prop:tolerance,rem:tolerance-midpoint}.

\end{document}